\PassOptionsToPackage{prologue,usenames,dvipsnames,table}{xcolor}
\PassOptionsToPackage{bookmarks,unicode,colorlinks=true}{hyperref}
\newif\ifanonymous

\anonymousfalse
\documentclass[runningheads]{llncs}

\usepackage{amssymb,amsmath,mathtools}
\usepackage{nicefrac}
\usepackage{mathrsfs}
\usepackage{xspace}
\usepackage[inline]{enumitem}
\usepackage{comment}
\usepackage{thm-restate}
\usepackage{graphicx}
\usepackage[colorlinks=true,linkcolor=violet,citecolor=blue]{hyperref}
\usepackage[capitalize,nameinlink]{cleveref}
\usepackage{adjustbox}
\usepackage{bm}
\usepackage{tikz}
\usepackage{ifthen}
\usetikzlibrary {arrows.meta}
\usetikzlibrary{backgrounds}
\usetikzlibrary{arrows,shapes}
\usetikzlibrary{tikzmark}
\usetikzlibrary{calc}
\usepackage{utfsym}
\usepackage{bussproofs}
\usepackage{latexsym} % More elegant symbol of eventually and always
\usepackage{tcolorbox}
\usepackage{algorithm}
\usepackage{algpseudocode}
\usepackage{booktabs}
\usepackage{caption}
\usepackage{wrapfig}
\usepackage{makecell}
\usepackage{multirow}
\usepackage{bbding}
\usepackage[misc]{ifsym}
\usepackage{orcidlink}

\spnewtheorem{subproblem}{Problem~1.\!\!}{\itshape}{}

\Crefname{problem}{Problem}{Problem}
\Crefname{subproblem}{Problem~1.\!\!}{Problem~1.\!\!}
\crefname{lemma}{Lem.}{Lem.}
\crefname{example}{Exmp.}{Exmp.}
\crefname{section}{Sect.}{Sect.}
\Crefname{appendix}{Appx.}{Appx.}
\crefname{definition}{Def.}{Def.}
\crefname{theorem}{Thm.}{Thm.}
\crefname{corollary}{Cor.}{Cor.}
\crefname{algorithm}{Alg.}{Alg.}

\input{macros.tex}
\input{tikzkey.tex}
\makeatletter
\def\thanks#1{\protected@xdef\@thanks{\@thanks
        \protect\footnotetext{#1}}}
\makeatother

\begin{document}
\title{Switching Controller Synthesis for Hybrid Systems Against STL Formulas\thanks{ Funding of this paper}}
%\thanks{Supported by organization x.}}
%
\titlerunning{Switching Controller Synthesis for HSs Against STL Formulas}
% If the paper title is too long for the running head, you can set
% an abbreviated paper title here
%

\ifanonymous
%    \author{Anonymous Author(s)}
    \author{}
    \institute{}
\else
    \author{Han Su\inst{1}\orcidlink{0000-0003-4260-8340} 
    \and
    Shenghua Feng\inst{2}\inst{(}\Envelope\inst{)}\orcidlink{0000-0002-5352-4954} 
    \and
    Sinong Zhan\inst{3}\orcidlink{0000-0002-5750-3296} \and
    Naijun Zhan\inst{4,1}\orcidlink{0000-0003-3298-3817}
    }
    \authorrunning{H.~Su et al.}

    \institute{Institute of Software, CAS, University of Chinese Academy of Sciences, Beijing, China 
    \email{\{suhan,znj\}@ios.ac.cn}
    \and
    Zhongguancun Laboratory, Beijing, China
    \email{fengsh@zgclab.edu.cn}\and
    Department of Electrical and Computer Engineering, Northwestern University
    \email{SinongZhan2028@u.northwestern.edu}\and
    School of Computer Science, Peking University, Beijing, China\\
    }
\fi
\maketitle              % typeset the header of the contribution

\begin{abstract}

Switching controllers play a pivotal role in directing hybrid systems (HSs) towards the desired objective, embodying a ``correct-by-construction'' approach to HS design. Identifying these objectives is thus crucial for the synthesis of effective switching controllers. While most of existing works focus on safety and liveness, few of them consider timing constraints. In this paper, we delves into the synthesis of switching controllers for HSs that meet system objectives given by a fragment of STL, which essentially corresponds to a reach-avoid problem with timing constraints. Our approach involves iteratively computing the state sets that can be driven to satisfy the reach-avoid specification with timing constraints. This technique supports to create switching controllers for both constant and non-constant HSs. We validate our method's soundness, and confirm its relative completeness for a certain subclass of HSs. Experiment results affirms the efficacy of our approach.

\end{abstract}
\section{Introduction}

    Hybrid systems (HSs) provide a robust mathematical specification in modeling cyber-physical systems (CPS) with their unique fusion of continuous physical dynamics and discrete switching behaviors. Many CPSs are often complex and safety-critical which necessitates intricate control specifications. Switching controller synthesis offers a formal guarantee of the given specification of HS. 
    %by refining the guard associated with each discrete transition and the domain constraint associated with each mode. 
    Its applications include attitude control in aerospace\cite{atkins2013aerospace}, aircraft collision-avoidance protocols in avionics\cite{tomlin2000game}, and pacemakers for treating bradycardia\cite{ye2008modelling}, etc.

%    Nowadays, some CPSs have pre-defined continuous behaviors for a limited variety of contexts, makes them generally incapable of overcoming changing objectives. One promising solution is switching controller synthesis, which gives a formal guarantee that the hybrid system fulfill a specification after system's continuous behavior is determind. Applications of this include the attitude control in aerospace\cite{atkins2013aerospace}, aircraft collision-avoidance protocols in avionics\cite{tomlin2000game}, and pacemakers for treating bradycardia\cite{ye2008modelling}. 
  
    With the escalating complexity of CPSs, the specifications required to ensure their proper functionality grow increasingly intricate. Among these, the importance of timing constraints becomes paramount. This is evident in various scenarios, from orchestrating synchronized reactions in chemical processing\cite{engell2000continuous} to ensuring seamless operations in multi-robot systems\cite{lindemann2019coupled}. In this context, Signal Temporal Logic (STL), a rigorous formalism for defining linear-time properties of continuous signals\cite{maler2004monitoring}, is exceptionally well-suited for specifying intricate timing constraints and qualitative properties of complex CPSs.

    However, switching controller synthesis for HSs \reviewercomment{double check to unify all the Hybrid Systems as HSs} against STL specifications is not well addressed in the literature.
    %\reviewercomment{Such argument seems too strong, as you cite a study on this topic later in related work. It is better to argue your merits relative to that work} 
    The primary challenge arises from the complex interactions between continuous behaviors and discrete transitions. 
    %The synthesized controllers must ensure the system adheres to both the timing constraints and qualitative properties specified  by STL, before and after discrete transitions. 
    A common technique to synthesize switching controllers for HSs with complex specifications is the abstraction-based method\cite{liu2013synthesis,mazo2010pessoa}. This technique involves abstracting the continuous state space of each mode into a finite set of states, which often results in the loss of precise timing information for each mode. 
    %However, it encounters limitations with STL due to its abstraction approach. 
    Consequently, the abstraction-based technique struggles with timing constraint analysis in the abstracted state space. In contrast, Mixed Integer Linear Programming (MILP) based technique \cite{raman2014model}
    for switching controller synthesis against STL specification can provide precise timing information, but this method faces challenges in handling the intricate interactions of diverse discrete transitions between modes.

    In this paper, we considered the switching controller synthesis problem for HSs against a fragment of STL specification, which essentially corresponds to a reach-avoid problem with timing constraints. To the best of our knowledge, this is the first work that uses STL to specify HSs with both discrete transitions and continuous dynamics. Similar work in \cite{da2021symbolic} focused only on HSs with discrete time dynamics in each mode, significantly simplifying the problem. The key idea behind our approach involves iteratively computing a sequence of state-time sets $(x,t)$, state $x$ and time $t$. These sets ensure that an HS, starting from state $x$ at time $t$, adheres to the STL specification within a certain number (i.e., the number of iterations) of switches. The state-time sets are computed explicitly when the dynamics of the HSs are constant, and are inner-approximated when the dynamics are non-constant. Based on the state-time sets, we propose a sound and relatively complete method to synthesize a switching controller that satisfies the STL specification. Our experimental results demonstrate the efficacy of this approach.%\nzcomment{Redundant, simplify the expression}

    The main contributions can be summarized as follows: 
    \begin{enumerate*}[label=(\roman*)]
        \item We conceptualize \emph{state-time set} for HSs. 
        %originating from where the executions can satisfy a given STL formula within a predetermined number of switches.
        \item We propose a methodology to synthesize switching controllers for HSs against a fragment of STL specification.
        \item We develop a prototype to demonstrate the efficiency and practical applicability of our methodology.
    \end{enumerate*} 

    \paragraph*{\textbf{Organization.}} \cref{sec:running-exp} gives an overview of our approach, \cref{sec:priliminary} provides a recap of important preliminaries and formally defines the problem. We illustrate the calculation of the state-time sets in \cref{sec:state-time}.  Based on the state-time sets, \cref{sec:prob1} shows how to derive switching systems against a STL specification. In \cref{sec:exp}, we demonstrate the efficacy of our method through several examples. 
    %\cref{subsec:add-usage} discusses the  an additional usage of state-time set. 
    We discuss related work in \cref{sec:related} and draw conclusion in \cref{sec:conclu}.
%    Abstraction based method can not work for STL, -- the time of the continuous behavior need to be considered.
%
%    Why add time into guard
%    \begin{itemize}
%        \item The syntax of STL naturally contain time;
%        \item Synthesis controller without time may result in a constrained system.
%    \end{itemize}
%
%    Example...
%    
%    Contribution...
%    
%    It is essential to consider time-dependent jump conditions when synthesizing controller for hybrid system with respect to a signal temporal logic specification. (illustrative example)
%
%    In some situations, the low-level plant controllers (such as PID controllers) have been designed before, but some new specifications are added to the CPS, then we can only modify the "cyber" part of the system to satisfying the given conditions. For instance...  
    
\section{An Illustrative Prelude}\label{sec:running-exp}

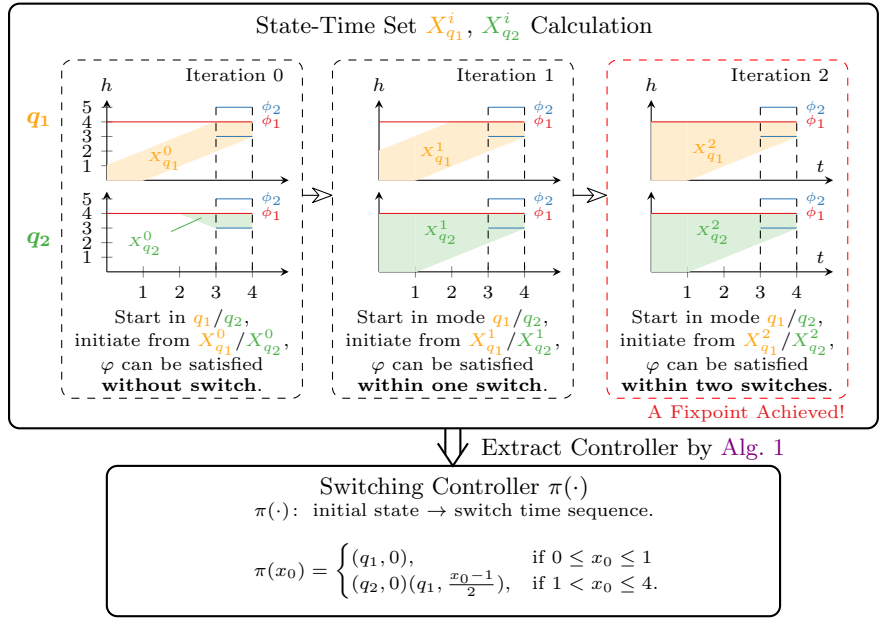
\begin{figure}[t]
		\vspace{-0.3cm}
        \centering
        \begin{adjustbox}{max width = 1\linewidth}
		    \begin{tikzpicture}[font=\scriptsize]
		        \begin{groupplot}[ Axis Set, width = 4cm , height = 2.65cm , group style={group size = 3 by 2, horizontal sep = 1.2cm, vertical sep = 0.15cm}]
		            %%%X_{q_1}^0
                    \nextgroupplot[xmax=5, ymax=5.5,ylabel=$h$,ytick={1,2,3,4,5},xtick=\empty]
                    % Guard conditions
                    \addplot[name path = G1U, domain = {0:1}, color = myorange!20] {x+1};
                    \addplot[name path = G1L, domain = {0:1}, color = myorange!20] {0};
                    \addplot[fill=myorange!20]fill between[of=G1U and G1L, soft clip={domain=0:1}];
                    \addplot[name path = G2U, domain = {1:3}, color = myorange!20] {x+1};
                    \addplot[name path = G2L, domain = {1:3}, color = myorange!20] {x-1};
                    \addplot[fill=myorange!20]fill between[of=G2U and G2L, soft clip={domain=1:3}];
                    \addplot[name path = G3U, domain = {3:4}, color = myorange!20] {4};
                    \addplot[name path = G3L, domain = {3:4}, color = myorange!20] {x-1};
                    \addplot[fill=myorange!20]fill between[of=G3U and G3L, soft clip={domain=3:4}];
                    \addplot[name path =pp, draw = none] (1.6,1.5) node {\orange{\tiny{$X_{q_1}^0$}}};
                    % The STL specification
		            \addplot[name path = Phi2Up,  domain = {3:4}, color=myblue] {5} node [right] {\tiny{$\phi_2$}};
		            \addplot[name path = Phi2Low, domain = {3:4}, color=myblue] {3};
                    \addplot[name path = Phi1, domain = {0:4}, color = myred] {4} node [right] {\tiny{$\phi_1$}};
                    %\addplot[fill=myblue!30]fill between[of=Phi2Up and Phi2Low, soft clip = {domain=3:4}];
                    % Dash line
		            \draw[dashed] ($(axis cs:3,\pgfkeysvalueof{/pgfplots/ymax})!{1/11}!(axis cs:3,\pgfkeysvalueof{/pgfplots/ymin})$) -- (axis cs:3,\pgfkeysvalueof{/pgfplots/ymin});
		            \draw[dashed] ($(axis cs:4,\pgfkeysvalueof{/pgfplots/ymax})!{1/11}!(axis cs:4,\pgfkeysvalueof{/pgfplots/ymin})$) -- (axis cs:4,\pgfkeysvalueof{/pgfplots/ymin});

                    %%% X_{q_1}^1 
                    \nextgroupplot[ymax=5.5, xmax=5, ylabel=$h$, ytick=\empty, xtick=\empty]
                    % Guard conditions
                    \addplot[name path = G1U, domain = {0:1}, color = myorange!20] {x+2};
                    \addplot[name path = G1L, domain = {0:1}, color = myorange!20] {0};
                    \addplot[fill=myorange!20]fill between[of=G1U and G1L, soft clip={domain=0:1}];
                    \addplot[name path = G2U, domain = {1:2}, color = myorange!20] {x+2};
                    \addplot[name path = G2L, domain = {1:2}, color = myorange!20] {x-1};
                    \addplot[fill=myorange!20]fill between[of=G2U and G2L, soft clip={domain=1:2}];
                    \addplot[name path = G3U, domain = {2:4}, color = myorange!20] {4};
                    \addplot[name path = G3L, domain = {2:4}, color = myorange!20] {x-1};
                    \addplot[fill=myorange!20]fill between[of=G3U and G3L, soft clip={domain=2:4}];
                    \addplot[name path =pp, draw = none] (1.6,1.8) node {\orange{\tiny{$X_{q_1}^1$}}};
                     % The STL specification
		            \addplot[name path = Phi2Up,  domain = {3:4}, color=myblue] {5} node [right] {\tiny{$\phi_2$}};
		            \addplot[name path = Phi2Low, domain = {3:4}, color=myblue] {3};
                    \addplot[name path = Phi1, domain = {0:4}, color = myred] {4} node [right] {\tiny{$\phi_1$}};
                    %\addplot[fill=myblue!30]fill between[of=Phi2Up and Phi2Low, soft clip = {domain=3:4}];
                    % Dash line
		            \draw[dashed] ($(axis cs:3,\pgfkeysvalueof{/pgfplots/ymax})!{1/11}!(axis cs:3,\pgfkeysvalueof{/pgfplots/ymin})$) -- (axis cs:3,\pgfkeysvalueof{/pgfplots/ymin});
		            \draw[dashed] ($(axis cs:4,\pgfkeysvalueof{/pgfplots/ymax})!{1/11}!(axis cs:4,\pgfkeysvalueof{/pgfplots/ymin})$) -- (axis cs:4,\pgfkeysvalueof{/pgfplots/ymin});

                    %%% X_{q_1}^2 
                    \nextgroupplot[ymax=5.5, xmax=5, xlabel=$t$, ylabel=$h$, ytick=\empty, xtick=\empty]
                    % Guard conditions
                    \addplot[name path = G1U, domain = {0:1}, color = myorange!20] {4};
                    \addplot[name path = G1L, domain = {0:1}, color = myorange!20] {0};
                    \addplot[fill=myorange!20]fill between[of=G1U and G1L, soft clip={domain=0:1}];
                    \addplot[name path = G2U, domain = {1:4}, color = myorange!20] {4};
                    \addplot[name path = G2L, domain = {1:4}, color = myorange!20] {x-1};
                    \addplot[fill=myorange!20]fill between[of=G2U and G2L, soft clip={domain=1:4}];
                    \addplot[name path =pp, draw = none] (1.6,2) node {\orange{\tiny{$X_{q_1}^2$}}};
                     % The STL specification
		            \addplot[name path = Phi2Up,  domain = {3:4}, color=myblue] {5} node [right] {\tiny{$\phi_2$}};
		            \addplot[name path = Phi2Low, domain = {3:4}, color=myblue] {3};
                    \addplot[name path = Phi1, domain = {0:4}, color = myred] {4} node [right] {\tiny{$\phi_1$}};
                    %\addplot[fill=myblue!30]fill between[of=Phi2Up and Phi2Low, soft clip = {domain=3:4}];
                    % Dash line
		            \draw[dashed] ($(axis cs:3,\pgfkeysvalueof{/pgfplots/ymax})!{1/11}!(axis cs:3,\pgfkeysvalueof{/pgfplots/ymin})$) -- (axis cs:3,\pgfkeysvalueof{/pgfplots/ymin});
		            \draw[dashed] ($(axis cs:4,\pgfkeysvalueof{/pgfplots/ymax})!{1/11}!(axis cs:4,\pgfkeysvalueof{/pgfplots/ymin})$) -- (axis cs:4,\pgfkeysvalueof{/pgfplots/ymin});

                    %%% X_{q_2}^0
                    \nextgroupplot[xmax=5, ymax=5.5,ytick={1,2,3,4,5}]
                    % Guard conditions
                    \addplot[name path = G1U, domain = {2:3}, color = mygreen!20] {4};
                    \addplot[name path = G1L, domain = {2:3}, color = mygreen!20] {-x+6};
                    \addplot[fill=mygreen!20]fill between[of=G1U and G1L, soft clip={domain=2:3}];
                    \addplot[name path = G2U, domain = {3:4}, color = mygreen!20] {4};
                    \addplot[name path = G2L, domain = {3:4}, color = mygreen!20] {3};
                    \addplot[fill=mygreen!20]fill between[of=G2U and G2L, soft clip={domain=3:4}];
                    \node[coordinate, pin={[pin edge={mygreen}]-170:{{\tiny{\green{$X_{q_2}^0$}}}}}] at (axis cs:2.6, 3.7){};
                    % The STL specification
		            \addplot[name path = Phi2Up,  domain = {3:4}, color=myblue] {5} node [right] {\tiny{$\phi_2$}};
		            \addplot[name path = Phi2Low, domain = {3:4}, color=myblue] {3};
                    \addplot[name path = Phi1, domain = {0:4}, color = myred] {4} node [right] {\tiny{$\phi_1$}};
                    %\addplot[fill=myblue!30]fill between[of=Phi2Up and Phi2Low, soft clip = {domain=3:4}];
                    % Dash line
		            \draw[dashed] ($(axis cs:3,\pgfkeysvalueof{/pgfplots/ymax})!{1/11}!(axis cs:3,\pgfkeysvalueof{/pgfplots/ymin})$) -- (axis cs:3,\pgfkeysvalueof{/pgfplots/ymin});
		            \draw[dashed] ($(axis cs:4,\pgfkeysvalueof{/pgfplots/ymax})!{1/11}!(axis cs:4,\pgfkeysvalueof{/pgfplots/ymin})$) -- (axis cs:4,\pgfkeysvalueof{/pgfplots/ymin});

                    %%% X_{q_2}^1
                    \nextgroupplot[ymax=5.5, xmax=5, ytick=\empty]
                    % Guard conditions
                    \addplot[name path = G1U, domain = {0:1}, color = mygreen!20] {4};
                    \addplot[name path = G1L, domain = {0:1}, color = mygreen!20] {0};
                    \addplot[fill=mygreen!20]fill between[of=G1U and G1L, soft clip={domain=0:1}];
                    \addplot[name path = G2U, domain = {1:4}, color = mygreen!20] {4};
                    \addplot[name path = G2L, domain = {1:4}, color = mygreen!20] {x-1};
                    \addplot[fill=mygreen!20]fill between[of=G2U and G2L, soft clip={domain=1:4}];
                    \addplot[name path =pp, draw = none] (1.7,2.7) node {\green{\tiny{$X_{q_2}^1$}}};
                     % The STL specification
		            \addplot[name path = Phi2Up,  domain = {3:4}, color=myblue] {5} node [right] {\tiny{$\phi_2$}};
		            \addplot[name path = Phi2Low, domain = {3:4}, color=myblue] {3};
                    \addplot[name path = Phi1, domain = {0:4}, color = myred] {4} node [right] {\tiny{$\phi_1$}};
                    %\addplot[fill=myblue!30]fill between[of=Phi2Up and Phi2Low, soft clip = {domain=3:4}];
                    % Dash line
		            \draw[dashed] ($(axis cs:3,\pgfkeysvalueof{/pgfplots/ymax})!{1/11}!(axis cs:3,\pgfkeysvalueof{/pgfplots/ymin})$) -- (axis cs:3,\pgfkeysvalueof{/pgfplots/ymin});
		            \draw[dashed] ($(axis cs:4,\pgfkeysvalueof{/pgfplots/ymax})!{1/11}!(axis cs:4,\pgfkeysvalueof{/pgfplots/ymin})$) -- (axis cs:4,\pgfkeysvalueof{/pgfplots/ymin});

                    %%% The possible trajectories from initial condition [3,4] in q2
                    \nextgroupplot[ymax=5.5, xmax=5, xlabel=$t$, ytick=\empty]
                    % Guard conditions
                    \addplot[name path = G1U, domain = {0:1}, color = mygreen!20] {4};
                    \addplot[name path = G1L, domain = {0:1}, color = mygreen!20] {0};
                    \addplot[fill=mygreen!20]fill between[of=G1U and G1L, soft clip={domain=0:1}];
                    \addplot[name path = G2U, domain = {1:4}, color = mygreen!20] {4};
                    \addplot[name path = G2L, domain = {1:4}, color = mygreen!20] {x-1};
                    \addplot[fill=mygreen!20]fill between[of=G2U and G2L, soft clip={domain=1:4}];
                    \addplot[name path =pp, draw = none] (1.7,2.7) node {\green{\tiny{$X_{q_2}^2$}}};
                     % The STL specification
		            \addplot[name path = Phi2Up,  domain = {3:4}, color=myblue] {5} node [right] {\tiny{$\phi_2$}};
		            \addplot[name path = Phi2Low, domain = {3:4}, color=myblue] {3};
                    \addplot[name path = Phi1, domain = {0:4}, color = myred] {4} node [right] {\tiny{$\phi_1$}};
                    %\addplot[fill=myblue!30]fill between[of=Phi2Up and Phi2Low, soft clip = {domain=3:4}];
                    % Dash line
		            \draw[dashed] ($(axis cs:3,\pgfkeysvalueof{/pgfplots/ymax})!{1/11}!(axis cs:3,\pgfkeysvalueof{/pgfplots/ymin})$) -- (axis cs:3,\pgfkeysvalueof{/pgfplots/ymin});
		            \draw[dashed] ($(axis cs:4,\pgfkeysvalueof{/pgfplots/ymax})!{1/11}!(axis cs:4,\pgfkeysvalueof{/pgfplots/ymin})$) -- (axis cs:4,\pgfkeysvalueof{/pgfplots/ymin});
		        \end{groupplot}        
                
                \draw[thick, rounded corners] (-1.3,2.35) |- (10.25,-3.3) -- (10.25,2.35) -- cycle;
                \node at(4.62,2.05) {\small{State-Time Set \orange{$X_{q_1}^i$}, \green{$X_{q_2}^i$} Calculation}}; 
                \node at(-0.9,0.8) {\small{\orange{$\bm{q_1}$}}};
                \node at(-0.9,-0.8) {\small{\green{$\bm{q_2}$}}};
                
                \draw[dashed,rounded corners] (-0.6,1.6) |- (2.6,-2.9) -- (2.6,1.6) -- cycle; 
                \node at(1.7,1.4) {Iteration 0};
                \node at(1,-2.3) {\scriptsize{\makecell[c]{Start in \orange{$q_1$}/\green{$q_2$},\\ 
                                                    initiate from \orange{$X_{q_1}^0$}/\green{$X_{q_2}^0$},\\ 
                                                    $\varphi$ can be satisfied\\ 
                                                    \textbf{without switch}.}}};
                \draw [-{Stealth[length=3mm, open, round]}] (2.6,-0.2) -- (3,-0.2);

                \draw[dashed, rounded corners] (3,1.6) |- (6.2,-2.9) -- (6.2,1.6) -- cycle;
                \node at(5.3,1.4) {Iteration 1};
                \node at(4.6,-2.3) {\scriptsize{\makecell[c]{Start in mode \orange{$q_1$}/\green{$q_2$},\\ 
                                                    initiate from \orange{$X_{q_1}^1$}/\green{$X_{q_2}^1$},\\ 
                                                    $\varphi$ can be satisfied\\ 
                                                    \textbf{within one switch}.}}};
                \draw [-{Stealth[length=3mm, open, round]}] (6.2,-0.2) -- (6.65,-0.2);
                
                \draw[dashed, rounded corners, color=red] (6.65,1.6) |- (9.85,-2.9) -- (9.85,1.6) -- cycle;
                \node at(8.95,1.4) {Iteration 2};
                \node at(8.25,-2.3) {\scriptsize{\makecell[c]{Start in mode \orange{$q_1$}/\green{$q_2$},\\ 
                                                    initiate from \orange{$X_{q_1}^2$}/\green{$X_{q_2}^2$},\\ 
                                                    $\varphi$ can be satisfied\\ 
                                                    \textbf{within two switches}.}}};
                \node at(8.5,-3.1) {\red{A Fixpoint Achieved!}};
                
                \draw[thick] (4.5,-3.3) -- (4.5,-3.65);
                \draw[thick] (4.7,-3.3) -- (4.7,-3.65);
                \draw[thick] (4.4,-3.5) -- (4.6,-3.8) -- (4.8,-3.5);
                \node at(7,-3.58) {\small{Extract Controller by \cref{alg:syn-ss}}};
                \begin{scope}
                    \draw[thick, rounded corners] (0,-3.8) |- (8.95,-5.8) -- (8.95,-3.8) -- cycle;
                    \node at(4.62,-4.1) {\small{Switching Controller $\pi(\cdot)$}};
                    %\node at(1.45,-4.3) {(Switching Time Controller $\pi(\cdot)$ )};
                    \node at(4.62,-4.4) {{$\pi(\cdot)\colon$ initial state $\to$ switch time sequence.}};
                    \node at(4.62,-5.2) {{
                    $\pi(x_0) = 
                    \begin{cases}
                        (q_1,0), & \text{if $0\leq x_0 \leq 1$}\\
                        (q_2,0) (q_1, \frac{x_0-1}{2}), &  \text{if $1< x_0 \leq 4$}.
                    \end{cases}$
                    }};
                \end{scope}
		    \end{tikzpicture}
        \end{adjustbox}
%        \begin{adjustbox}{max width = 0.5\linewidth}
%            \begin{tikzpicture}[oriented WD, align=center, bbx=1.3cm, bby=.5cm,font=\scriptsize]
%                    \node[bb={1}{1}, bb min width=.9in] (q1) {\bm{$q_1:$} K1 ON, K2 OFF\\$\Flow[q_1]:\dot x=1$\\$\Init[q_1]:0\le x\le 3$\\$\Inv[q_1]:0\le x\le 6$};
%                    \node[bb={1}{1}, bb min width=.9in, below = .8cm of q1] (q2) {\bm{$q_2:$} K1 OFF, K2 ON\\$\Flow[q_1]:\dot x=-1$\\$\Init[q_1]:3\le x\le 6$\\$\Inv[q_2]:0\le x\le 6$};
%                    \draw[ar] let \p1=(q1.west), \p2=(q2.west), \n1={3*\bbportlen} in (q1_in1) -- (\x1-\n1,\y1) -- node[fill=white, inner sep=1pt] {$e_1=(q_1,q_2)$\\$\Jump[e_1]:x\ge 4\wedge x=x'$} (\x2-\n1,\y2) -- (q2_in1);
%                    \draw[ar] let \p1=(q1.east), \p2=(q2.east), \n1={3*\bbportlen} in (q2_out1) -- (\x2+\n1,\y2)  -- node[fill=white,inner sep=0.8pt] {$e_2={q_2,q_1}$\\$\Jump[e_2]:x\le 1\wedge x=x'$} (\x1+\n1,\y1) -- (q1_out1);
%            \end{tikzpicture}
%        \end{adjustbox}
%
        \caption{An overview of our method}
        \label{fig:overview}
    \end{figure}
  % We use the following example to give an illustrative prelude of our approach.
    \begin{example}\label{exp:WT-init}
        In the reactor system depicted in \cref{fig:WaterTank}, liquid is continuously consumed by the reaction and is replenished through pipe $P$. The system alternates
        {\makeatletter
        \let\par\@@par
        \par\parshape0
        \everypar{}
        
        \begin{wrapfigure}{r}{0.34\textwidth}
        \centering
        \vspace{-0.6cm}
        \begin{adjustbox}{max width = 0.6\linewidth}
            \begin{tikzpicture}[font=\scriptsize,scale=1,
                pattern1/.style={draw=black,pattern color=black!60,pattern=north east lines}]
                \begin{scope}
                    %\draw[help lines] (0,0) grid (4,4);
                    %\draw[thick] {[thick,rounded corners] (0,3) |- (2.5,0)} -- (2.5,0.3) -- (2.6,0.3) -- (2.6,0.5) --  (2.5,0.5) -- (2.5,3); 
                    \draw[thick] {[thick,rounded corners] (0,3) |- (2.5,0)} -- (2.5,3); 
                    \draw[thick] (0,2) -- (2.5,2); % Water Tank
                    \fill[mygray,opacity=0.3]   {[thick,rounded corners] (0,2) |- (2.5,0)} -- (2.5,2) -- cycle;  % Water level 
                    \draw[thick] (2,3) |- (2.3,1.5) -- (2.3,3) -- (2,3);
                    \fill[pattern1, thick] (2,3) |- (2.3,1.5) -- (2.3,3) -- cycle; 
                    \node at(1.8,2.7) {$R$};
                    % Valve K2
                    %\draw[thick,fill=white]  (2.6,0.25) -- (2.8,0.55) -- (2.8,0.25) -- (2.6,0.55) -- cycle; 
                    %\node at (3,0.7) {$K_2$};
                    %\draw[thick] (2.9,0.3) -- (2.8,0.3) -- (2.8,0.5) -- (2.9,0.5);
                    % Valve K1
                    \draw[thick,fill=white] (0.2,3.1) -- (0.5,3.3) -- (0.2,3.3) -- (0.5,3.1) -- cycle;
                    \draw[thick] (0.25,3) -- (0.25,3.1) -- (0.45,3.1) -- (0.45,3);
                    \draw[thick] (0.25,3.3) --++ (0,0.2) arc (0:90:0.2) --++ (-0.05,0);
                    \draw[thick] (0.45,3.3) --++ (0,0.2) arc (0:90:0.4) --++ (-0.05,0);
                    \fill[mygray,opacity=0.3] (0.25,3.3) --++ (0,0.2) arc (0:90:0.2) --++ (-0.05,0) --++ (0,0.2) --++ (0.05,0) arc (90:0:0.4) --++ (0,-0.2) -- cycle; 
                    \node at (0.6,3.6) {$P$};
                \end{scope}
                \draw[|<->|] (0.3,0) -- node[right]{$h$}(0.3,2);
                \draw[|<->|] (-0.2,0) -- node[fill=white,inner ysep=2pt,inner xsep=0]{$5$}(-0.2,3);
latexthe key                 \draw[|<->|] (2.1,0) -- node[right]{$3$}(2.1,1.5);
                %\begin{scope}
                %    \node (controller) at(2.7,3.5) [draw,align=center,thick] {Control \\ Algorithm};
                %    \draw (controller.south) -- (2.7,0.4) node [sloped,midway,above] {Off\textbackslash On};
                %    \draw (controller.west) -| (0.8,3.2) -- (0.35,3.2);
                %    \draw (1.35,3.7) node {On\textbackslash Off}; 
                %\end{scope}
                \node at(1.2,-0.7) {\small{$\begin{array}{l}
                                        \bm{q_1}:\dot h = 1\\
                                        \bm{q_2}:\dot h = -1
                                            \end{array}$}};
            \end{tikzpicture}
        \end{adjustbox}
        \vspace{-0.3cm}
        \caption{Reactor System}
        \label{fig:WaterTank}
    \end{wrapfigure}
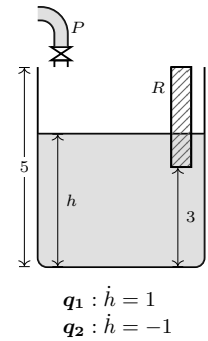
    \noindent
     between modes of adding liquid ($\bm{q_1}$) and exclusively consuming it ($\bm{q_2}$). The objectives are to keep the liquid level, $h$, between 0 and 4 meters, and to ensure that $h$ remains between 3 and 5 meters at a certain point during a critical reaction phase - time interval 3 to 4, for proper interaction with the reactor rod $R$. These objectives can be given as an STL formula $\varphi = (0\leq h\leq 4)\until[3,4](3\leq h\leq 5)$. 
    \qedT

    We present the core idea behind our approach in \cref{fig:overview}. Initially, we compute the state-time set $X_q^i$ iteratively. As shown in the upper block of \cref{fig:overview}, the set $X_q^i$ encompasses states in mode $q$ from which $\varphi$ can be satisfied within $i$ switches
    \par}%
    \noindent
      (as detailed in \cref{sec:prob1}). Once these \emph{state-time sets} are determined, the  switching controllers can be synthesized using the methods outlined in \cref{alg:syn-ss} and \cref{alg:time-dependent controller}.

    \end{example}  

\section{Notations and Problem Formulation}\label{sec:priliminary}
    \paragraph{Notations.} Let $\Nats,\Reals,$ and $\NonNegReals $ denote the set of natural, real, and non-negative real numbers, respectively. 
    %\haocomment{$\PosReals$, positive or nonnegative?}
    %\hscomment{nonnegative, correct it}
    %Let $\Realn$ denotes $n$-dimensional real space. 
    Given vector $x \in \Realn$, $x_i$ refers to its $i$-th component, and $p\,[x\repl u ]$ denotes the replacement of $x$ by $u$ for any predicate $p$ where $x$ serves as a variable. 
    %And $p\geq 0$ is an atomic predicate with $(p\geq 0)\,[x\repl u ]$ defined by $p\,[x\repl u] \geq 0$. For any set $S$, $S[x_i \repl a]$ denotes the projection of $S$ over $x_i = a$, that is $S[x_i \repl a] \defeq \{ (x_1,\dots,x_{i-1}, x_{i+1},\dots,x_n) \mid  (x_1,\dots,x_{i-1},a, x_{i+1},\dots,x_n) \in S\}$. 
    
    %Given any non-negative interval $I=[l,u]\subseteq\NonNegReals$ and $t\in\NonNegReals$, $I\Pminus t \defeq [l - t, u -t] \cap \NonNegReals$ denote the positive part of $I$ after minus $t$.
    %\hscomment{We do not use bold fonts to represent vectors here to avoid symbol confusion later}

    \paragraph{Differential dynamics.} We consider a class of dynamical systems featuring differential dynamics governed by ordinary differential equations (ODEs) of the form $\dot{\xx} = f(\xx)$, where $f$ is a continuous differentiable function. Given an initial state $x_0\in \Realn$, there exist a unique solution $\xx:\Reals_{\geq 0} \to \Realn$ in the sense that $\dot{\xx}(t) = f(\xx(t))$ for $t\geq 0$ and $\xx(0) = x_0$. 
    %Let $$\Seman{f} \defeq \{\xx \mid \xx \text{ is a solution of } f\}$$ denote the set of all trajectories that follow vector field $f$.Here we slightly abuse the notation of $\xx$, as it can both represent a state in $\Realn$ or a function from time to state.

   \paragraph{Switched systems.} A switched system is defined as a tuple $\Phi = (Q, F, \Init, \pi)$, where 
   \begin{itemize}
       \item $Q \defeq \{\,q_1, q_2,\dots, q_m\,\}$ is a finite set of discrete modes.
       \item $F\defeq \{\, f_q\mid q\in Q\,\}$ is a set of vector fields, and each mode $q\in Q$ endows with a unique vector field $f_q$ which specifies how system evolves in mode $q$.
       \item $\Init \subseteq \Realn$ is a set of initial states.
       \item $\pi\colon \Init \to (\Reals_{\geq 0} \to Q)$ is a switching controller. The controller maps each initial state $x_0\in \Init$ to a piecewise constant function $\pi(x_0)$, which in turn maps a time $t$ to the corresponding control mode $\pi(x_0)(t)$. 
   \end{itemize}
   
   % Switched systems are a class of hybrid systems that consists of a finite set of discrete modes $Q \defeq \{q_1, q_2,\dots, q_m\}$, a set of vector fields $F\defeq \{ f_q\mid q\in Q\}$, initial states $I\subseteq \Realn$, and a switching time controller $\pi\colon I \to (\Reals_{\geq 0} \to Q)$ that maps each initial state $x_0\in I$ to a piecewise constant function $\pi(x_0)$. 
   
   Given any initial state $x_0$, the dynamics of the switched system $\Phi$ is governed by equation $ \dot{\xx} (t) = f_{\pi(x_0)(t)}(\xx(t))$ with initial condition $\xx(0) = x_0$.

\paragraph{Signal temporal reach-avoid.} We consider a fragment of signal temporal logic, namely \emph{signal temporal reach-avoid} formula (ST-RA for short). The syntax of ST-RA is defined by
    \begin{align*}
        &\phi\Def~ \mu(x,t) \ge 0 ~|~ \neg\phi ~|~ \phi\wedge\phi\\
        &\varphi\Def ~ \phi_1 \,\mathcal{U}_I\, \phi_2
    \end{align*}
where $\phi$ is a Boolean combination of predicates over $x$ and time $t$, $I \defeq [l,u]$ is a closed time interval for some $0\leq l\leq u$. Intuitively, an ST-RA formula $\varphi$ expresses the requirement that the system should reach $\phi_2$ while avoid leaving $\phi_1$ within time frame $I$. The semantics of ST-RA formula, in alignment with STL, is defined as the satisfaction of a formula $\varphi$ with respect to a signal $\xx$ and a time instant $t$.
%\hscomment{Discuss how to support the whole STL}
\begin{remark}
Compared with the standard signal temporal logic (STL)~\cite{maler2004monitoring}, ST-RA formula does not allow nested ``until" operator, this makes ST-RA formula a fragment of STL. %\rev{However, ST-RA have already described a large variety of specifications in CPS, such as battery charge management and traffic light control.}
\end{remark}

Formally, given function $\xx\from \Reals_{\geq 0}\to \Realn$ (termed signal) and time $\tau$, the satisfaction of $\varphi$ at $(\xx, \tau)$, denoted by $(\xx,\tau) \models \varphi$, is inductively defined as follows:
\begin{align*}
            &(\xx,\tau)\models \mu(x,t)\ge 0 & & \text{iff} \quad \mu(\xx(\tau),\tau)\ge 0\,;\\
            &(\xx,\tau)\models \neg\phi & & \text{iff} \quad (\xx,\tau)\not\models \phi\,;\\  
            &(\xx,\tau)\models \phi_1\wedge\phi_2 & & \text{iff} \quad (\xx,\tau)\models \phi_1 \text{ and } (\xx,\tau)\models \phi_2\,;\\  
            &(\xx,\tau)\models \phi_1 \,\until[I]\, \phi_2 & & \text{iff} \quad \exists \tau'\ge \tau, \text{ such that } \tau'-\tau\in I,\, (\xx, \tau') \models \phi_2\, ,\\
            &&& \qquad \qquad \quad \;\; \text{and } \forall \tau''\in [\tau,\tau'], (\xx,\tau'')\models\phi_1 \,.
\end{align*}
Intuitively, the subscript $I$ in the until operator $\mathcal{U}_I$ defines the timing constraints under which a signal must reach $\phi_2$ while avoid leaving $\phi_1$. % We may omit, without conflicting with the standard until semantic in temporal logic, the subscript $I$ in $\phi_1 \,\mathcal{U}_I \,\phi_2$ when $I = [0,\infty)$.

Given a ST-RA formula $\phi$, we say a switched system $\Phi = (Q,F, \Init, \pi)$ models $\varphi$, denoted by $\Phi \models \varphi$, if $(\xx, 0) \models \varphi$ for any trajectory $\xx$ starting from initial set $\Init$. %Similarly, we say a switched hybrid automata $\HyAuto = (Q, F, \Init, \Dom, E, G)$ models $\varphi$, denoted by $\HyAuto \models \varphi$, if $(\xx, 0) \models \varphi$ for any trajectory $(q(t), \xx(t))$ starting from initial set $\Init$. 
We now formulate the problem addressed in this paper. 

\paragraph{\textbf{Problem Formulation.}}\label{para:prob_form}
Suppose there exists a finite set of control modes $Q = \{\, q_1, q_2, \dots, q_m\,\}$ and associated vector fields $F = \{\, f_{q_1}, f_{q_2}, \dots, f_{q_m}\,\}$. Each mode $q\in Q$ is associated with a vector field $f_q$ that governs the system's behavior in mode $q$. Let $\varphi = \phi_1 \,\mathcal{U}_I\,\phi_2$ be an ST-RA formula, a natural question is how to design a system that incorporates mode $q\in Q$ as subsystems, while ensuring any trajectory of the system satisfies $\varphi$. To address this, we formulate the problem as follows.

%As discussed in remark~\ref{remark:comp_switched_and_automata}, the switching mechanism varies in different types of hybrid systems. Thus, if time-dependent switching mechanism is preferred, our goal is to synthesize a switched system.
%\sfcomment{needs polish}
%\hscomment{Perhaps we could emphasize in problem 1 that our goal is to achieve a switched system with the minimum number of switches required to satisfy $\varphi$.}
%\sfcomment{This can be state in the next section as theoretic guarantee of our method.}
\begin{tcolorbox}[boxrule=.5pt,colback=white,colframe=black!75]
\textbf{Synthesis of Switched System.} Given a finite set of discrete modes $Q = \{\, q_1,q_2,\dots,q_m\,\}$, a set of vector fields $F = \{\, f_{q_1}, f_{q_2}, \dots, f_{q_m}\,\}$, and a ST-RA formula $\varphi = \phi_1 \,\mathcal{U}_I \,\phi_2 $, the switched system synthesis problem  aims to synthesize a switched system $\Phi = (Q, F, \Init , \pi)$, such that $\Phi \models \varphi$.
\end{tcolorbox}
%
%
%Similarly, if state-dependent switching mechanism is favored, our objective shifts to the synthesis of a switched hybrid automaton.
%
%\begin{tcolorbox}[boxrule=.5pt,colback=white,colframe=black!75]
%\textbf{Synthesis of Switched Hybrid Automata.} Given a finite set of discrete modes $Q = \{\, q_1,q_2,\dots,q_m\,\}$, a set of vector fields $F = \{\, f_{q_1}, f_{q_2}, \dots, f_{q_m}\,\}$, a set of edges $E\subseteq Q\times Q$, and a ST-RA formula $\varphi = \phi_1 \,\mathcal{U}_I \,\phi_2 $, the switched hybrid automata synthesis problem  aims to synthesize a switched hybrid automaton $\HyAuto =(Q, F, \Init, \Dom, E, G)$, such that $\HyAuto \models \varphi$.
%\end{tcolorbox}
%\hscomment{Add $E$ as a given condition}
%    \subsection{Problem Statement}
%        \begin{problem}\label{prob:origi-def}
%            Given a hybrid automaton $\HyAuto$ and an STL formula $\varphi$, redefine the time-dependent jump conditions of $\HyAuto$, such that $\forall\Traj[]\in\llangle\HyAuto\rrangle$, $\Traj[]\vDash\varphi$.   
%        \end{problem}
        %\begin{tcolorbox}[boxrule=1pt,colback=white,colframe=black!75]
        %    f
        %\end{tcolorbox}

\begin{remark}
The solutions to the above synthesis problem is inherently non-unique and may encompass trivialities, such as the one only with an empty initial set. Therefore, our goal is to identify a system with a nontrivial initial set $\Init$.  % In fact, a more larger initial set $\Init$ provides more probability to synthesize a nontrivial switching controller against   $\varphi$.
%\simoncomment{last sentence is bit confusing here}
\end{remark}

%\begin{remark}
%    \rev{In this paper, while our primary focus has been on synthesizing switched systems, the hierarchy we propose—particularly the method for calculating the state-time set—may be applicable to the synthesis of other types of controllers. This potential extension is discussed briefly in \cref{subsec:add-usage}.}
%    \nzcomment{We may delete this remark}
%    \sfcomment{I agree.}
%\end{remark}

%\begin{remark}
%The formulation of the problems does not explicitly impose a restriction on the state space. However, this apparent lack of restriction does not present a limitation, as any required restriction can be integrated into the specification $\varphi$. Specifically, if there is an additional state space constraint $S$, the task effectively becomes identifying a system that satisfies a modified requirement $\varphi' = \varphi \wedge (x \in S)$, which integrates the original specification $\varphi$ with the state space constraint, denoted by $(x \in S)$.
%\end{remark} 

%\sfcomment{Todo, summary of section 4, 5} In section XXX, we propose ... to switched system synthesize problem, base on result in section, we ,,, for the switched hybrid automata synthesis problem. 
\section{State-Time Set and its Calculation}\label{sec:state-time}

This section dedicates to synthesize a switched system $\Phi$ %(or equivalently, a time-dependent controller) 
that satisfies the given ST-RA formula  $\varphi = \phi_1\, \mathcal{U}_I\, \phi_2$. The key idea behind our approach is to compute a sequence of state-time sets $\{X_q^i\}_{q\in Q}$ for $i\in \mathbb{N}$, where $X_q^i$ denotes the set of all $(x,\tau)$ such that starting from $x$ at time $\tau$ in mode $q$, the system can be driven to reach $\phi_2$ while satisfying $\phi_1$ \emph{within $i$ times of switches}. In what follows, we first formally propose the concept of state-time sets and show how to calculate it explicitly. Subsequently, leveraging these state-time sets, we demonstrate the synthesis of a switched system that satisfies $\varphi$ in \cref{sec:prob1}.

\subsection{State-Time Sets}
The concept of state-time set is formally captured by the following definition.
\begin{definition} [State-time sets]\label{def:state-time_set}
For any $i\in\Nats$ and any $q\in Q$,  let $X_q^i$ denote the set of all state-time pairs $(x,\tau)$ such that there exists a controller $\pi(x) \from [\tau,\infty )\to Q$, satisfying 
  \begin{enumerate}[label=(\roman*)]
        \item $\pi(x)(\tau) = q$, and the piecewise constant function $\pi(x)$ contains at most $i$ discontinuous points; \label{cond1:switching_time}%\hscomment{It may not be appropriate to use $i+1$ different modes (see the example, only 2 modes exists, some state-time set may include 2 switches, thus needs to contain 3 different modes following this definition).} \sfcomment{Sure, change to at most $i$ discontinuous points}
        \item $(\xx, \tau) \models \phi_1\until_{I\Pminus \tau}\phi_2 $, where $\xx$ is the solution of ODE $\dot{\xx}(t) = f_{\pi(x)(t) } (\xx(t), t) $ over $[\tau,\infty)$ with $\xx(\tau) = x$,  and $I\Pminus \tau \defeq [l - \tau, u -\tau] \cap \NonNegReals$ for any interval $I = [l,u]$. \label{cond2:reach}
  \end{enumerate}
\end{definition}
Intuitively, condition~\ref{cond2:reach} suggests that the system can be driven to reach $\phi_2$ while satisfying $\phi_1$ from $x$ at time $\tau$, and condition~\ref{cond1:switching_time} indicates that the system initially remains in mode $q$, and the switching controller undergoes no more than $i$ switches. 
From the above definition of state-time sets, the following results can be derived:

\begin{restatable}{corollary}
{restateAllInitial}
\label{cor:all_initial} 
The following properties hold for the state-time sets $\{X_q^i\}_{q\in Q}$:
\begin{enumerate}[parsep = 1ex]
    \item For any $q\in Q$, $\{X_q^i\}$ is monotonically increasing, i.e. $X_q^0 \subseteq X_q^1 \subseteq X_q^2 \subseteq \cdots $.
    \item For any $i\in \Nats$ and  any $x \in  X_q^i [t \repl 0]$, $x$ can be driven to satisfy $\phi_1\until_{I}\phi_2$, i.e. there exists a switching controller $\pi$, such that 
$(\xx,0)\models \phi_1\until_{I}\phi_2 $, where $\xx$ is the trajectory starting from $x$ at time $0$ under controller $\pi$, and $X_i^q [t \repl 0] \defeq \{x \mid (x,0)\in X_i^q\}$ is the projection of $X_i^q$ into $t = 0$.
    \item  $\displaystyle \cup_{i\in \Nats} \cup_{q\in Q} X_q^i[t\repl 0]$ is the set of all states that can be driven to satisfy $\phi_1\until_{I}\phi_2$.
\end{enumerate}
\end{restatable}

%\begin{proof}
%    Trivial. \qed
%\end{proof}

According to \cref{cor:all_initial}, the state-time sets encompass the initial set of the switched system that we intend to synthesize. However, the state-time set and controller defined in \cref{def:state-time_set} are not given explicitly. To address this, we first elucidate the process of calculating the state-time sets. %Following this, we will detail the method for extracting a controller from these sets. 

The subsequent result establishes a relationship between the sets $\{X_q^i\}_{q\in Q}$ and $\{X_{q}^{i-1}\}_{q\in Q}$, forming the foundation for the inductive computation of state-time sets.%\reviewercomment{ Discuss the connection of our approach to the backward induction approach in model checking for temporal logic (LTL)} This connection is pivotal in understanding the progressive evolution of these sets and serves as a critical step in our analytical approach.

\begin{restatable}{theorem}
{restateInducStateTimeSet}
\label{thm:induc_state-time_set}
Follow the notations as before, we have\footnote{For any $a, b \in \PosReals$ such that $a \leq b$, the constraint $a \leq t \leq b$ is concisely denoted as $t \in [a, b]$.}
\begin{enumerate}
    \item Given any $q\in Q$, $(x,\tau)\in X_q^0$ if and only if 
    \begin{equation} \label{eq:induc_base}
        (\xx, \tau) \models \phi_1\, \mathcal{U}\, (\phi_2\wedge (t\in I))
    \end{equation}
    where $\xx$ is the solution of ODE $\dot{\xx}(t) = f_{q } (\xx(t), t) $ over $[\tau,\infty)$ with $\xx(\tau) = x$.
    \item Given any $q\in Q$, for any $i\geq 1$, $(x,\tau)\in X_{q}^i$ if and only if 
    \begin{equation} \label{eq:induc_i}
        \exists q'\neq q\in Q, \;(\xx,\tau) \models  \phi_1\, \mathcal{U}\, X_{q'}^{i-1}
    \end{equation}
    where $\xx$ is the solution of ODE $\dot{\xx}(t) = f_{q } (\xx(t), t) $ over $[\tau,\infty)$ with $\xx(\tau) = x$.
\end{enumerate}
\end{restatable}

\oomit{\begin{proof}[Proof sketch]
    \qed
\end{proof}}

For any formula $\psi(u,v)$, let  $\qe\left(\exists u,\, \psi(u,v)\right) \defeq \{ v\mid \exists u,\, \text{s.t. } \psi(u, v)\text{ holds} \} $  denote the set of all $v$ for which $\exists u,\, \psi(u,v)$ is true. Utilizing this notation, the state-time sets can be represented inductively.
\begin{restatable}{theorem}
{restateComputeTimeStateSet}
\label{thm:compute_time-state_set}
    For any $q\in Q$, suppose the solution of ODE $ \dot{\xx}(t) = f_q (\xx(t)) $ with initial $x$ at time $\tau$ 
    %can be explicitly solved
    is denoted by $\Psi(\,\cdot\, ; x, \tau, q)$, then the state-time sets can be inductively represented by
    %\reviewercomment{ Detailed the type of solution, i.e., polynomial solution or exponential solution}\hscomment{I believe we can regard Theorem 2 as a general result for all solvable dynamic system, and we may further discuss the type of solution in the new subsection to be added rather than in here}
    \begin{align}
    &X_q^{0}= \qe \left(  \exists \delta \ge 0,~ \Big(\phi_2[(x,t)\repl(\Psi(t+ \delta;x, t, q),t+\delta)] \wedge (t + \delta \in I)\Big) \right. \label{eq:X_q^0} \\
        &\qquad \qquad \qquad \qquad \qquad \wedge \left.\Big( \forall 0\leq h\leq \delta,\, \phi_1 [(x,t)\repl(\Psi(t+ h;x, t, q),t+h)]\Big)  \right) \notag\\
       & X_q^{i}=\bigvee_{q'\neq q} \qe \left(  \exists \delta \ge 0,~\left( X_{q'}^{i-1}[(x,t)\repl(\Psi(t+ \delta;x, t, q),t+\delta)]\right)\right. \label{eq:X_q^i}\\
        &\qquad \qquad \qquad \qquad \qquad \wedge \left.\Big( \forall 0\leq h\leq \delta,\, \phi_1 [(x,t)\repl(\Psi(t+ h;x, t, q),t+h)]\Big)  \right) \notag
    \end{align}
    for any $q\in Q$ and any $i\in \Nats$.
\end{restatable}
%\ooit{
%\begin{proof}
%    Trivial. \qed
%\end{proof} }

\begin{remark}
    When $\psi(u,v)$ consists of a Boolean combination of polynomial inequalities, a decidable procedure, such as cylindrical algebraic decomposition~\cite{arnon1984cylindrical}, exists for computing $\qe\left(\exists u,\, \psi(u,v)\right)$. This procedure exhibits a complexity that is double exponential with respect to the number of variables involved.
\end{remark}

\begin{remark}
    Our methodology essentially shares the idea of backward induction in controller synthesis for timed games \cite{de2003element}. 
    However, our approach diverges in two key aspects: (1) the safety/target sets and timing constraints are intricately interwoven in ST-RA formula, necessitating their concurrent consideration at each step of the induction process; (2) our method operates within an infinite-dimensional space due to the continuous nature of the state space, in contrast to the backward induction for timed games, which is confined to a finite set of k-polyhedra.
\end{remark}

%\rev{
%Our methodology essentially shares the idea backward induction in controller synthesis for timed games \cite{de2003element}. However, applying this technique to hybrid systems against ST-RA presents significant challenges due to the diverse dynamics of hybrid systems, which substantially increase the complexity of backward induction at each step. For some certain types of ODE, it may be undecidable to explicitly calculate the state-time set. In the following subsection, we will address this challenge and explore extensions of \cref{thm:induc_state-time_set} to ODEs that may lack analytic solutions.
%}

\subsection{Computing/Approximating State-Time Sets}

Although \cref{thm:compute_time-state_set} offers an inductive representation of $X_q^i$, the explicit computation of \cref{eq:X_q^0,eq:X_q^i} are challenging in general. This difficulty arises from two main factors: 
\begin{enumerate*}[label=(\roman*)]
    \item the necessity to explicitly solve the ordinary differential equation in each mode, and

    \item the high complexity of $\qe$, and the potential inclusion of non-elementary functions (such as exponential functions) in \cref{eq:X_q^0,eq:X_q^i}, for which a generally decidable procedure to solve $\qe$ may not exist.
\end{enumerate*} 

To address the difficulties outlined above, we categorize the dynamics into constant and non-constant systems. For the constant dynamics, its solution can be directly computed, and there exists a decidable procedure to solve $\qe$ with a complexity polynomially dependent on the formula length. For the non-constant dynamics, due to their high complexity, we forego an explicit solution for the state-time set and instead demonstrate a method to approximate this set.

\paragraph{Constant dynamics.} Suppose the dynamics within each mode $q \in Q$ is constant, and both $\phi_1$ and $\phi_2$ are Boolean combinations of linear inequalities, 
\cref{eq:X_q^0,eq:X_q^i} can be effectively solved using readily available solvers, such as Z3~\cite{de2008z3}. \cref{thm:compute_time-state_set} directly implies the following result.

\begin{restatable}{corollary}
{restateConstantFlow}
\label{cor:constant_flow}
    Following the notations as before, suppose the dynamics within each mode is constant, i.e. $f_q = a_q$ for any $q\in Q$, and both $\phi_1$ and $\phi_2$ are Boolean combinations of linear inequalities, then $\{X_q^i\}_{q\in Q}$ can be inductively solved by \cref{eq:X_q^0,eq:X_q^i} with $ \Psi(t; x, \tau, q) = x + (t-\tau)\cdot a_q$.
\end{restatable}

\begin{remark}
    Although $\qe$ on polynomial constraints is double-exponential in general~\cite{arnon1984cylindrical}, constant dynamics facilitate a relatively efficient (\emph{polynomial in formula length}) solving procedure. This comes from the following observation: (1) the $\qe$ procedure in \cref{eq:X_q^0,eq:X_q^i} operates in polynomial time when the constraints are linear and involve only a single existential and a single universal variable~\cite[Thm~6.2]{weispfenning1988complexity} . (2) if $X_{q}^{i-1}$ is linear for all $q\in Q$, then $X_{q}^i$ is also linear.
\end{remark}

We now illustrate the computation process of $X_q^i$ via the following example.

\begin{example}\label{ex:state-time_set_cal}
Let's reconsider the reactor system in \cref{exp:WT-init}. The reactor system consists of two modes $q_1$ and $q_2$ with $f_{q_1} = 1$, $f_{q_2} = -1$, and the liquid level requirement is $\varphi = (0\leq h\leq 4)\until[3,4](3\leq h\leq 5)$. 

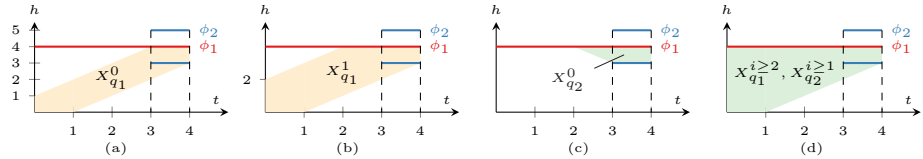
\begin{figure}[h]
\vspace{-0.2cm}
		    \centering
            \begin{adjustbox}{max width = 1\linewidth}
		        \begin{tikzpicture}[font=\tiny]
		            \begin{groupplot}[ Axis Set, width = 4.2cm, height = 2.8cm, group style={group size = 4 by 1, horizontal sep = 0.5cm}]
		                %%%X_{q_1}^1
                        \nextgroupplot[xmax=5, ymax=5.5,xlabel=$t$,ylabel=$h$,ytick={1,2,3,4,5}]
                        \addplot[name path = X1, draw = none] (2,2.1) node {$X_{q_1}^0$};
                        \addplot[name path = G1U, domain = {0:1}, draw = none] {x+1};
                        \addplot[name path = G1L, domain = {0:1}, draw = none] {0};
                        \addplot[fill=myorange!20]fill between[of=G1U and G1L, soft clip={domain=0:1}];
                        \addplot[name path = G2U, domain = {1:3}, draw = none] {x+1};
                        \addplot[name path = G2L, domain = {1:3}, draw = none] {x-1};
                        \addplot[fill=myorange!20]fill between[of=G2U and G2L, soft clip={domain=1:3}];
                        \addplot[name path = G3U, domain = {3:4}, draw = none] {4};
                        \addplot[name path = G3L, domain = {3:4}, draw = none] {x-1};
                        \addplot[fill=myorange!20]fill between[of=G3U and G3L, soft clip={domain=3:4}];
                        % The STL specification
		                \addplot[name path = Phi2Up,  domain = {3:4}, color=myblue, thick] {5} node [right] {$\phi_2$};
		                \addplot[name path = Phi2Low, domain = {3:4}, color=myblue, thick] {3};
                        \addplot[name path = Phi1, domain = {0:4}, color = myred, thick] {4} node [right] {$\phi_1$};
                        % Dash line
		                \draw[dashed] ($(axis cs:3,\pgfkeysvalueof{/pgfplots/ymax})!{1/11}!(axis cs:3,\pgfkeysvalueof{/pgfplots/ymin})$) -- (axis cs:3,\pgfkeysvalueof{/pgfplots/ymin});
		                \draw[dashed] ($(axis cs:4,\pgfkeysvalueof{/pgfplots/ymax})!{1/11}!(axis cs:4,\pgfkeysvalueof{/pgfplots/ymin})$) -- (axis cs:4,\pgfkeysvalueof{/pgfplots/ymin});

                        %%% X_{q_2}^0
                        \nextgroupplot[ymax=5.5, xmax=5, xlabel=$t$, ylabel=$h$, ytick={2}]
                        \addplot[name path = X1, draw = none] (2,2.4) node {$X_{q_1}^1$};
                        \addplot[name path = G1L, domain = {0:1}, draw = none] {0};
                        \addplot[name path = G1U, domain = {0:1}, draw = none] {x+2};
                        \addplot[fill=myorange!20]fill between[of=G1U and G1L, soft clip={domain=0:1}];
                        \addplot[name path = G2L, domain = {1:2}, draw = none] {x-1};
                        \addplot[name path = G2U, domain = {1:2}, draw = none] {x+2};
                        \addplot[fill=myorange!20]fill between[of=G2U and G2L, soft clip={domain=1:2}];
                        \addplot[name path = G2L, domain = {2:4}, draw = none] {x-1};
                        \addplot[name path = G2U, domain = {2:4}, draw = none] {4};
                        \addplot[fill=myorange!20]fill between[of=G2U and G2L, soft clip={domain=2:4}];
                        % The STL Specification
		                \addplot[name path=Phi2Up,domain={3:4},color=myblue,thick] {5} node [right] {$\phi_2$};
		                \addplot[name path=Phi2Low,domain={3:4},color=myblue,thick] {3};
                        \addplot[name path=Phi1, domain={0:4},color=myred,thick] {4} node [right] {$\phi_1$};
                        % Dashed lines
		                \draw[dashed] ($(axis cs:3,\pgfkeysvalueof{/pgfplots/ymax})!{1/11}!(axis cs:3,\pgfkeysvalueof{/pgfplots/ymin})$) -- (axis cs:3,\pgfkeysvalueof{/pgfplots/ymin});
		                \draw[dashed] ($(axis cs:4,\pgfkeysvalueof{/pgfplots/ymax})!{1/11}!(axis cs:4,\pgfkeysvalueof{/pgfplots/ymin})$) -- (axis cs:4,\pgfkeysvalueof{/pgfplots/ymin});

                        \nextgroupplot[ymax=5.5, xmax=5, xlabel=$t$, ylabel=$h$, ytick = \empty]
                        %\addplot[name path = X1, draw = none] (2,2.4) node {$X_{q_2}^0$};
                        \node[coordinate, pin={[pin edge={Black}]-170:{{\color{Black}$X_{q_2}^0$}}}] at (axis cs:3.3, 3.5){};
                        \addplot[name path = G1L, domain = {2:3}, draw = none] {-x+6};
                        \addplot[name path = G1U, domain = {2:3}, draw = none] {4};
                        \addplot[fill=mygreen!20]fill between[of=G1U and G1L, soft clip={domain=2:3}];
                        \addplot[name path = G2L, domain = {3:4}, draw = none] {3};
                        \addplot[name path = G2U, domain = {3:4}, draw = none] {4};
                        \addplot[fill=mygreen!20]fill between[of=G2U and G2L, soft clip={domain=3:4}];
                        % The STL Specification
		                \addplot[name path=Phi2Up,domain={3:4},color=myblue,thick] {5} node [right] {$\phi_2$};
		                \addplot[name path=Phi2Low,domain={3:4},color=myblue,thick] {3};
                        \addplot[name path=Phi1, domain={0:4},color=myred,thick] {4} node [right] {$\phi_1$};
                        % Dashed lines
		                \draw[dashed] ($(axis cs:3,\pgfkeysvalueof{/pgfplots/ymax})!{1/11}!(axis cs:3,\pgfkeysvalueof{/pgfplots/ymin})$) -- (axis cs:3,\pgfkeysvalueof{/pgfplots/ymin});
		                \draw[dashed] ($(axis cs:4,\pgfkeysvalueof{/pgfplots/ymax})!{1/11}!(axis cs:4,\pgfkeysvalueof{/pgfplots/ymin})$) -- (axis cs:4,\pgfkeysvalueof{/pgfplots/ymin});

                        \nextgroupplot[ymax=5.5, xmax=5, xlabel=$t$, ylabel=$h$, ytick = \empty]
                        \addplot[name path = X1, draw = none] (1.5,2.4) node {$X_{q_1}^{i\ge 2}, X_{q_2}^{i\ge 1}$};
                        \addplot[name path = G1L, domain = {0:1}, draw = none] {0};
                        \addplot[name path = G1U, domain = {0:1}, draw = none] {4};
                        \addplot[fill=mygreen!20]fill between[of=G1U and G1L, soft clip={domain=0:1}];
                        \addplot[name path = G2L, domain = {1:4}, draw = none] {x-1};
                        \addplot[name path = G2U, domain = {1:4}, draw = none] {4};
                        \addplot[fill=mygreen!20]fill between[of=G2U and G2L, soft clip={domain=1:4}];
                        % The STL Specification
		                \addplot[name path=Phi2Up,domain={3:4},color=myblue,thick] {5} node [right] {$\phi_2$};
		                \addplot[name path=Phi2Low,domain={3:4},color=myblue,thick] {3};
                        \addplot[name path=Phi1, domain={0:4},color=myred,thick] {4} node [right] {$\phi_1$};
                        % Dashed lines
		                \draw[dashed] ($(axis cs:3,\pgfkeysvalueof{/pgfplots/ymax})!{1/11}!(axis cs:3,\pgfkeysvalueof{/pgfplots/ymin})$) -- (axis cs:3,\pgfkeysvalueof{/pgfplots/ymin});
		                \draw[dashed] ($(axis cs:4,\pgfkeysvalueof{/pgfplots/ymax})!{1/11}!(axis cs:4,\pgfkeysvalueof{/pgfplots/ymin})$) -- (axis cs:4,\pgfkeysvalueof{/pgfplots/ymin});
                        
		            \end{groupplot}

                    \node at(1.1, -0.5) {(a)};
                    \node at(4.23, -0.5) {(b)};
                    \node at(7.35, -0.5) {(c)};
                    \node at(10.5, -0.5) {(d)};
		        \end{tikzpicture}
            \end{adjustbox}
            \vspace{-0.7cm}
            \caption{The state-time sets calculation of the reactor system in \cref{exp:WT-init}. The state-time sets reach a fixpoint after $2$ iteration.}
            \label{fig:state-time_set}
            \vspace{-.1cm}
            \end{figure}
%\noindent Based on \cref{eq:X_q^0,eq:X_q^i}, the state-time sets we calculate are illustrated in \cref{fig:state-time_set}. The procedure reaches a fixpoint within 2 iterations for any $q\in Q$.   
%have 
%\begin{align*}
%    & X_{q_1}^0 = \{ (h,t) \mid (t-1 \le h \le t+1) \wedge (0\le h \le 4)\wedge (0\le t\le 4)\},\\
%    & X_{q_2}^0 = \{ (h,t) \mid (6-t \le h ) \wedge (3 \le h\le 4) \wedge (0\le t \le 4)\}.
%\end{align*}
%for $i = 0$, and 
%\begin{align*}
%    & X_{q_1}^1 = \{ (h,t) \mid (t-1 \le h \le t+2) \wedge (0\le h\le 4) \wedge (0\le t\le 4)   \},\\
%    & X_{q_2}^1 = \{ (h,t) \mid (t-1 \le h \le 4) \wedge (0\le t\le 4) \}.
%\end{align*}
%for $i = 1$, and 
%\begin{align*}
%    & X_{q_1}^2 = \{ (h,t) \mid  (t-1 \le h\le 4) \wedge (0\le t \le 4)  \},\\
%    & X_{q_2}^2 = \{ (h,t) \mid  (t-1 \le h\le 4) \wedge (0\le t \le 4)\}.
%\end{align*}
%for $i=2$. Continuing the process, we can find $X_q^i = X_q^{i+1}$ for $i\geq 2$, indicating that $X_q^i$ %attains a fixpoint within 2 iterations for any $q\in Q$.  \qedT 
Based on \cref{eq:X_q^0,eq:X_q^i}, the state-time sets we calculate are illustrated in \cref{fig:state-time_set}. The procedure reaches a fixpoint within 2 iterations for any $q\in Q$.   \qedT

\end{example}

\paragraph{Non-constant dynamics.} Assuming that the dynamics are non-constant, the exact computation of  \cref{eq:X_q^0,eq:X_q^i} may prove to be overly complex or potentially undecidable. We thus seek to \emph{inner-approximate} the state-time sets. According to \cref{thm:induc_state-time_set}, 
\begin{itemize}
    \item[-] $X_q^0$ is the set from which the system in mode $q$ will satisfy $\phi_1\, \mathcal{U}\, (\phi_2\wedge (t\in I))$;
    \item[-] $X_q^i$ is the set from which the system in mode $q$ will satisfy $\phi_1\, \mathcal{U}\, X_{q'}^{i-1}$ for some $q'\in Q$.
\end{itemize}
We identify that the crucial element for inner-approximating the state-time sets lies in employing a method that finds sets from which the system will satisfy a classical `until' or `reach-avoid' formula\footnote{This problem is also referred to as the inner approximation of the reach-avoid problem}. Numerous studies have explored this issue; in this paper, we employ the approach proposed in~\cite{xue2023reach}.

\begin{theorem}[Inner-approximation of Reach-avoid Set~\cite{xue2023reach}]\label{thm:result_in_xue}
    Given dynamic system $\dot \xx(t) = f(\xx(t))$, safety set $\psi_1 \subseteq\Realn$ and target set $ \psi_2 \subseteq \Realn$. If there exists continuously differentiable function $v(x):\overline{\psi_1}\to\Reals$ and $w(x):\overline{\psi_1} \to\Reals$, satisfying \footnote{$\nabla_x v(x)$ represents the gradient of $v(x)$ with respect to $x$, $\overline{\psi_1}$ denotes the closure of set $\psi_1$ and $\partial \psi_1$ refers to the boundary of $\psi_1$.}
\begin{align*}
    \left\{\begin{aligned}
        & \nabla_{x} v(x) \cdot f(x) \ge 0, \quad \forall x\in \overline{\psi_1\setminus \psi_2},\\
        & v(x) - \nabla_{x} w(x)\cdot f(x) \le 0, \quad\forall x \in \overline{\psi_1\setminus \psi_2},\\
        & v(x) \le 0,\quad\forall x\in\partial \psi_1,
    \end{aligned}\right.
\end{align*}
then any trajectory starting from $\{x\mid v(x)\ge 0\}$ satisfies formula $\psi_1\, \mathcal{U}\,\psi_2$.
\end{theorem}
\begin{remark}
    The synthesis of function $v(x)$ and  $w(x)$ can be reduce to a SDP problem. For a detailed formulation, we refer the reader to \cite{xue2023reach}.
\end{remark}

Since the state-time sets $X_q^{i}$ depend on both the state $x$ and time $t$, we first lift the dynamics of each mode to a higher dimension that incorporates time $t$. Specifically, the dynamics in mode $q$ are transformed into $\left(\dot{x}, \dot{t}\right) = (f_q,1)$. Subsequently, employing \cref{thm:result_in_xue} and \cref{thm:induc_state-time_set}, we can inductively inner-approximate the state-time set $X_q^i$. The resulting approximation is denoted by $\widetilde{X_q^i}$.

\begin{example}\label{exp:appro}
    
Consider a temperature control system featuring two modes, $q_1$ and $q_2$, with dynamics given by $f_{q_1} = 20 - 0.2x - 0.001x^2$ and $f_{q_2} = -0.2x - 0.001x^2$, 
{\makeatletter
	\let\par\@@par
	\par\parshape0
	\everypar{}
\begin{wrapfigure}{r}{0.48\textwidth}
        \vspace{-0.8cm}
       \centering
       \begin{adjustbox}{max width = 1.1\linewidth}
           \scalebox{1.4}{
               \begin{tikzpicture}
                \node at (0,0) {\includegraphics[width=1\textwidth]{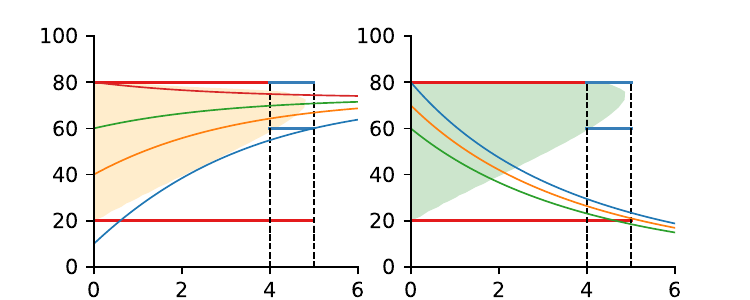}};
                \draw[thick, -{Stealth[length=3mm, width=2mm]}] (-4.573,1.7) -- (-4.573,2.2); % stealth风格的箭头
                \draw[thick, -{Stealth[length=3mm, width=2mm]}] (0.583,1.7) -- (0.583,2.2); % stealth风格的箭头
                \draw[thick, -{Stealth[length=3mm, width=2mm]}] (-0.5,-1.9) -- (0,-1.9); % stealth风格的箭头
                \draw[thick, -{Stealth[length=3mm, width=2mm]}] (4.7,-1.9) -- (5.2,-1.9); % stealth风格的箭头
                \node at(-0.2,-1.6) {{$t$}};
                \node at(5,-1.6) {$t$};
                \node at(-4.3,2) {$x$};
                \node at(0.9,2) {$x$};

                \node at(-3,0.5) {\large{$\widetilde{X_{q_1}^0}$}};
                \node at(2.5,0.5) {\large{$\widetilde{X_{q_2}^1}$}};

                \node at(-2.8,-2.8) {\large{(a)}};
                \node at(2.5,-2.8) {\large{(b)}};
               \end{tikzpicture}
          }
       \end{adjustbox}
       \vspace{-0.7cm}
    \caption{The state-time sets approximation of temperature controller system}
      \label{fig:approx}
  \end{wrapfigure}
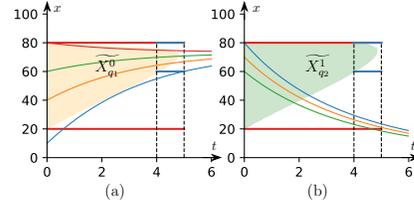
    \noindent
    where $x$ represents the temperature. The control objective is defined by the ST-RA formula $\varphi = (20\le x\le 80) \until_{[4,5]} (60\le x\le80)$.
 
    \cref{fig:approx} presents the result obtained by inner-approximating\footnote{The approximation of $X_{q_2}^0$ and $X_{q_1}^1$ is an empty set, hence it is not depicted.} $X_{q_1}^0$, $X_{q_2}^0$, $X_{q_1}^1$, and $X_{q_2}^0$. Based on the results, we observe that when $x$ is within the range of $[20, 80]$ in mode $q_1$, the system can satisfy $\varphi$ without any switching. However, 
    \par}
    \noindent
    for $x\in [20, 80]$ in mode $q_2$, at least one switch is necessary for $\varphi$ to be satisfied.\qedT
\end{example}

\section{Synthesizing Switched Systems}\label{sec:prob1}
In this section, we demonstrate the synthesis of a switched system $\Phi$ that conforms to the formula $\varphi = \phi_1, \mathcal{U}_I, \phi_2$. This synthesis builds on the state-time sets introduced in Section \ref{sec:state-time}. We initially outline the synthesis procedure for the switched system in \cref{alg:syn-ss} and subsequently describe the extraction of a switching controller in \cref{alg:time-dependent controller}.

\paragraph{Switched System Synthesis.}
We now summary the synthesis algorithm in \cref{alg:syn-ss}. Given any $k\in \Nats$ that serves as a prescribed upper bound of switching time, \cref{alg:syn-ss} inductively calculates/inner-approximates\footnote{To clarify, we continue to use $X_q^i$ to represent the inner approximation of the state-time sets, rather than using $\widetilde{X_q^i}$.} state-time sets $\{X_q^i\}_{q\in Q}$ (line \ref{alg1:X_q^0}, \ref{alg1:X_q^i}), and partition $X_q^k[t\repl 0]$ into $\Init(q)^i\defeq (X_q^i \setminus X_q^{i-1})[t\repl 0]$ (line \ref{alg1:init0}, \ref{alg1:initi}) for $i = 0,1, \dots, k$. $\Init(q)^i$ denote the set of states (in mode $q$) that can be driven to satisfy $\varphi$ with \emph{at least $i$ times} of switching (cf. \cref{cor:all_initial}).
%synthesizes a switched hybrid system satisfying requirement $\varphi$.
%Based on \cref{alg:time-relevant-syn}, \cref{alg:syn-ss} synthesizes a switched hybrid system satisfying requirement $\varphi$. 
The initial set is defined by 
$$\Init = \cup_{q\in Q}\cup_{i=0}^k \Init[q]^i,$$ 
which contains states that can be driven to satisfy $\varphi$ within $k$ times of switching, and the switching controller $\pi$ is synthesized by \cref{alg:time-dependent controller} (line \ref{alg1:controller}). 

    \begin{algorithm}[t]
    %\vspace{-1cm}
        \caption{Synthesis of Switched system}
        \label{alg:syn-ss}
        \begin{algorithmic}[1] % The number tells where the line numbering should start
            \Require  $Q$, $F$, $\varphi=\phi_1\until[I]\phi_2$, and $k$ 
            \Comment{$k$ is the upper bound of switching time}
            \Ensure A switched system $\Phi = (Q,F,\Init,\pi)$, such that $\Phi\models \varphi$
                %\State Call \cref{alg:time-relevant-syn} to obtain  $\{X_q^i\}_{i=0}^k$ and $\{\Init[q]^i\}_{i=0}^k$ for $q\in Q$
                                \ForAll{$q\in\DisState$} \label{alg1:init-start}
                    \State $X_q^0\gets $ inner-approximate/explicitly calculate $X_q^0$ \label{alg1:X_q^0}
                    \State $\Init[q]^0\gets X_q^0[t\repl 0]$ \label{alg1:init0}
                \EndFor \label{alg1:init-end}
                \For{$i = 1,2,\cdots,k$} \label{alg1:induct-start}
                    \ForAll{$q\in Q$}
                        \State $X_q^i\gets$ inner-approximate/explicitly calculate $X_q^i$ \label{alg1:X_q^i}
                        \State $\Init[q]^i\gets (X_q^i \setminus X_q^{i-1})[t\repl 0]$ \Comment{$\Init[q]^i$ is recorded for controller synthesis}\label{alg1:initi}
                    \EndFor
                    %\State $i\gets i+1$
                \EndFor
                \State $\Init \gets \cup_{q\in Q}\cup_{i=0}^k \Init[q]^i$       \Comment{Initial set}
    
                \State Call \cref{alg:time-dependent controller} to obtain controller $\pi$ \label{alg1:controller}
                \Comment{
                \begin{minipage}[t]{5.7cm}
                Given any $x_0\in \Init$, \cref{alg:time-dependent controller} 
                computes the controller that drives $x_0$ to satisfy $\varphi$
                \end{minipage}
                }
            %\EndProcedure
        \end{algorithmic}
    \end{algorithm}

\paragraph{Switching controller synthesis.} 
For any $x_0\in \Init$, \cref{alg:time-dependent controller} computes the controller that drives $x_0$ to satisfy $\varphi$. 
\cref{alg:time-dependent controller} first finds $\Init[q_0]^l$ that contains $x_0$ with $l$ be the smallest index (line \ref{alg:init_mode}). $l$ is the smallest switching time that can drive $x_0$ to satisfy $\varphi$, and the subscript $q_0$ indicates $x_0$ first lies in mode $q_0$. 

Line \ref{alg:start}--\ref{alg:end} find the next switching time and switching mode. 
Let $\Reach(t; x_0, t_0, q)$ denote the over-approximation of the reachable set starting from $(x_0, t_0)$ in mode $q$ at time $t$. 
%The existence of $q_j$ in line \ref{alg:next_switching_time} is ensured by the construction of the state-time sets, snd 
Next switching time $\widetilde{t}$ and switching mode $\widetilde{q}$ are chosen to ensure that the system enters $X_{\widetilde{q}}^{l-j}$ at time $\widetilde{t}$ in mode $q_{j-1}$, this is formally encoded by 
\[\Reach( \widetilde{t} ;x_{j-1},t_{j-1},q_{i-1}) \subseteq X_{\widetilde{q}}^{l-j}[t = \widetilde{t}\,].\]

In line~\ref{alg:output_controller}, the controller $\pi$ maps $x_0$ to a piecewise constant function $\pi(x_0) = (q_0,t_0)(q_1,t_1)\cdots(q_l,t_l) $, which represents a function that maps $t$ to $q_i$ if $t_{i} \leq t < t_{i+1}$.

% In line~\ref{alg:output_controller}, the controller $\pi$ assigns the initial state $x_0$ to a piecewise constant function. Specifically, $\pi(x_0) = (q_0, t_0)(q_1, t_1)\cdots(q_l, t_l)$, where each $q_i$ is mapped to the time interval $t_i \leq t < t_{i+1}$.

\begin{algorithm}[t]
    \caption{Switching controller synthesis}
    \label{alg:time-dependent controller}
    \begin{algorithmic}[1] % The number tells where the line numbering should start
        \Require  $x_0$, $\{X_q^i\}_{i=0}^k$, and $\{\Init[q]^i\}_{i=0}^k$ \Comment{$x_0$ is the initial state}
        \Ensure  $\pi(x_0)$ \Comment{The switching controller}
        %\Procedure{$SynSwit$}{$x_0, \{\Init[q]^i\}_{i=0}^k,\{G_t(e)^i\}_{i=1}^k$ } \Comment{$x_0$ is the initial state}
            %\State $\{\Init[q]^i\}_{i=0}^k, \{G_t(e)^i\}_{i=1}^k \gets StateTimeCal(Q,F,E,\varphi,k)$
            %\State $I \gets \cup_{q\in Q}\cup_{i\le k} \Init[q]^i$
            \State Find the initial set $\Init[q_0]^l$ that includes $x_0$ and has the smallest index $l$ \label{alg:init_mode}
            \State Select $q_0$ as initial mode, $t_0 \gets 0$                    
            \For{$j=1,\cdots,l$}
                %\State Select $t_{j}\in G_t(q_{j-1},q_{j})^{i-j+1}[x \repl \xx_{q_{j-1}}(t;x_{j-1},t_{j-1})]$
                \For{$q\in Q$} \label{alg:start}
                    \If{$\Reach( \widetilde{t} ;x_{j-1},t_{j-1},q_{i-1}) \subseteq X_{\widetilde{q}}^{l-j}[t = \widetilde{t}\,]$ for some $\widetilde{t} > t_{j-1}$,\,$\widetilde{q}\in Q$}
                        \State Select $t_j\gets \widetilde{t}$,~ $q_j\gets \widetilde{q}$ \label{alg:next_switching_time}
                        \State $x_j\gets \Reach(t_j;x_{j-1},t_{j-1},q_{j-1})$\label{alg:next_switching_state}
                        \State \textbf{Break}
                    \EndIf
                \EndFor \label{alg:end}
                % \State Select a $q_j\in Q$ such that $X_{q_j}^{l-j}[x\repl \Psi(t;x_{j-1}, t_{j-1}, q_{j-1})]$ is nonempty \label{alg:next_switching_time}
                % \State Select $t_{j} \in X_{q_j}^{l-j}[x\repl \Psi(t;x_{j-1}, t_{j-1}, q_{j-1})]$ 
                %\State Select $t_{j} \in X_{q_j}^{l-j}[x\repl p(t;x_{j-1}, t_{j-1}, q_{j-1})\!+\!I]$ 
                % \State $x_{j} \gets \Psi(t_j;x_{j-1},t_{j-1},q_{j-1})$ \label{alg:next_switching_state}
                %\State $x_{j} \gets p(t_j;x_{j-1},t_{j-1},q_{j-1})\! +\! I$ \label{alg:next_switching_state}
            \EndFor
            \State $\pi(x_0) = (q_0,t_0)(q_1,t_1)\cdots(q_l,t_l) $ \Comment{
            \begin{minipage}[t]{6cm}
                Representing a piecewise constant function such that $\pi(x_0)(t) = q_i$ if $t_{i} \leq t < t_{i+1}$
            \end{minipage}
            } \label{alg:output_controller}
        %\EndProcedure
    \end{algorithmic}
\end{algorithm}
\begin{remark}
    Numerous methods are available to estimate the reachable set of a dynamic system \cite{chen2012taylor,xue2016under,xue2019inner}. In this paper, we employ Flow* \cite{chen2013flow}, a method based on Taylor model, to over-approximate the reachable set. 
\end{remark}

\begin{remark}
    Assuming that the dynamics (i.e. $f_q$) within each mode remain constant, the reachable set can be explicitly calculated. This, in conjunction with the explicit calculation of state-time sets, is crucial for demonstrating relative completeness in the context of constant dynamics (c.f. \cref{thm:theoretic}). 
\end{remark}

\begin{remark}
    For non-constant dynamics, since the state-time sets and reachable sets are inner- and over-approximated, there may exist an initial state $x_0$ that can be driven to satisfy the ST-RA formula, while our method fails to identify a controller.
\end{remark}
%\paragraph{Approximating Reachable Sets.}
%    Over-approximating the state of a system at time $t$ with a given initial condition frees us from analyzing the entire state space when synthesizing a switching controller. We only need to focus on the state within the reachable set, thereby reducing the conservation of the synthesized controller. Specially, if the reachable set originating from $(x_0,t_0)\in X_q^i$ is enclosed by a state-time set $X_{q'}^{i-1}$ at time $t$, the HS can switch from $q$ to $q'$ at time $t$ while ensuring compliance with the ST-RA formula.
%    \sfcomment{hard to understand}

%    Numerous methods are available to estimate the reachable set of a dynamic system \cite{chen2012taylor,xue2016under,xue2019inner}. In this paper, we employ Flow* \cite{chen2013flow}, a method based on Taylor model, to over-approximate the reachable set. We define $\Reach(t; x_0, t_0, q)$ as the 

We now illustrate our approach through two examples.

\begin{example}\label{exp:alg1-resul}
    In \cref{ex:state-time_set_cal}, we have obtained the state-time sets $\{X_{q_1}^i,X_{q_2}^i\}$ for $i \leq 2$, thus, according to \cref{alg:syn-ss} (with $k=2$), we have
    \begin{align*}
                &\Init[q_1]^0 = [0,1], & &\Init[q_1]^1 = (1,2], && \Init[q_1]^2 = (2,4]\\
                &\Init[q_2]^0 = \emptyset, && \Init[q_2]^1 = [0,4], && \Init[q_2]^2 = \emptyset.
            \end{align*} 
    Based on this, we can synthesize a switched system $\Phi$ with $\Init = \{h \mid 0\leq h\leq 4\}$. The corresponding switching controller $\pi$ is defined by %\sfcomment{check whether the controller is consistent with the overview}
                \[
                \pi(x_0) = 
                \begin{cases}
                    (q_1,0), & \text{if $0\leq x_0 \leq 1$}\\
                    ( q_2,0) (q_1, \frac{x_0-1}{2}), &  \text{if $1< x_0 \leq 4$}.
                \end{cases}
            \tag*{\qedT}\] 
            
\end{example}

 {\makeatletter
	\let\par\@@par
	\par\parshape0
	\everypar{}
       
    \noindent 
%    reachable set from $(x_0, t_0)$ at time $t$, where the dynamics are described by $f_q$. The following example offers a clear demonstration of using reachable sets to synthe-
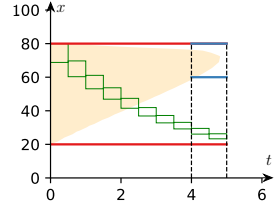
\begin{wrapfigure}{r}{0.34\textwidth}
    \vspace{-0.2cm}        
    \centering
            \begin{adjustbox}{max width = 0.92\linewidth}
                \scalebox{1}{
                    \begin{tikzpicture}
                        \node at(0,0) {\includegraphics[width=1\textwidth]{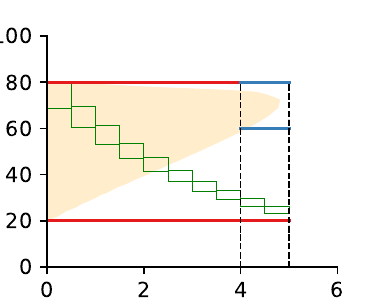}};
                        \draw[very thick, -{Stealth[length=5mm, width=4mm]}] (4.9,-3.8) -- (5.4,-3.8); % stealth风格的箭
                        \draw[very thick, -{Stealth[length=5mm, width=4mm]}] (-4.57,3.6) -- (-4.57,4.1); % stealth风格的箭
                        \node at(-6.1,3.7) {\huge{\textsf{1}}};
                        \node at(5.2,-3) {\huge{$t$}};
                        \node at(-4.1,3.8) {\huge{$x$}};
                    \end{tikzpicture}
               }
            \end{adjustbox}
            \vspace{-0.1cm}
            \caption{Switching controller synthesis of \cref{exp:appro}}
           \label{fig:appro-syn}
        \end{wrapfigure}

        \begin{example}\label{exp:appro-syn}
        Let's reconsider \cref{exp:appro}, we demonstrate our approach by synthesizing the switching controller for initial state $x_0=80$ in mode $q_2$. The reachable set $\Reach(t;x_0,t_0,q_2)$ is represented by green boxes in \cref{fig:appro-syn}. We observe the reachable set will enter $X_{q_1}^0$ for any $t\in [0,2]$, this implies initial state $x_0=80$ in mode $q_2$ can be driven to satisfy $\varphi$ if the system switches into mode $q_1$ within time interval $[0,2]$, i.e. $\pi(80) = (q_2,0) (q_1,\widetilde{t})$ for any $\widetilde{t} \in [0,2]$.
        \qedT
        \end{example}
    \par}%
    \noindent
    
    The following result states the advantages of our approach.
%, namely: soundness, relative completeness, and minimal switching property.
    \begin{restatable}[Soundness, Relative Completeness, Minimal Switching Property]{theorem}
    {restateTheoretic}\label{thm:theoretic}
        Given modes $Q$, vector fields $F$, and formula $\varphi=\phi_1\,\mathcal{U}_I\, \phi_2$, the following results hold:
        \begin{enumerate}
            \item \cref{alg:syn-ss} is sound, that is $\Phi\models \varphi$;
            \item \cref{alg:syn-ss} is relatively complete for constant dynamics:
            \reviewercomment{The name "relative completeness" may evoke an unrelated theorem for dynamic system}
           for any $x\in \Realn$, if $x$ can be driven to satisfy $\varphi$ with some controller $\pi$, then there exists $k \in \Nats$\footnote{In fact, $k$ can be chosen to the number of discontinuous points of $\pi(x)$. } , such that the initial set of the synthesized switched system contains $x$.
            \item The controller synthesized in \cref{alg:time-dependent controller} features minimal switching property for constant dynamics: for any $x_0\in \Init$, there does not exists any controller $\pi'$, that can drive $x_0$ to satisfy $\varphi$ with switching time (equivalently, number of discontinuous points of $\pi'(x_0)$) less than $\pi(x_0)$.
        \end{enumerate}
    \end{restatable}

    \begin{remark}
        Suppose the dynamic in each mode can be explicitly solved and there exists a decidable procedure for solving $\qe(\cdot)$ in \cref{eq:X_q^0,eq:X_q^i}, then \cref{alg:syn-ss} is also relatively complete and the corresponding controller also features minimal switching property.
    \end{remark}
\section{Experimental Evaluation}\label{sec:exp}

%\hscommentinline{1. Polish the expression of this section (Scalable)}
%\hscommentinline{2. Add example about general hybrid system}
%\hscommentinline{3. Rearrange the table - dimension,fixed-point}

    We develop a prototype\footnote{Available at \url{https://github.com/Han-SU/BenchMark_STLControlSyn4HS}} of our synthesis method in Python, employing the Z3 solver \cite{de2008z3} to explicitly compute the state-time sets for HSs with constant dynamics. For HSs with linear or polynomial dynamics, we use the semidefinite programming solver MOSEK \cite{mosek} to approximate the state-time set. The prototype is evaluated on various benchmark examples using a laptop with a 3.49GHz Apple M2 processor, 8GB RAM, and macOS 14.3.

    \begin{table*}[t]
        \centering
            %\vspace{-1cm}
            \captionsetup{font={scriptsize}}
            \caption{ST-RA Specifications}
            \vspace{-0.3cm}
            \label{tab:explain}
            \begin{center}
                \begin{tabular}{l  c l } 
                    \toprule
                    Model  & ~ & ST-RA Formulas \\
                    \midrule
                    \textsf{Reactor}\cite{zhao2013synthesizing}
                        & ~ & \scriptsize{\makecell[l]{$
                            \varphi~ : (10\!\le\! \textit{tempe} \!\le\! 90)\wedge(0\!\le\!\textit{cooling}\le\! 1) \until[15,20](40\!\le\!\textit{tempe} \!\le\! 50)$} }\\
                    \cmidrule{2-3}
                    \multirow{3}{*}{\textsf{WaterTank}\cite{raisch1999approximating}} 
                        & ~ & \scriptsize{\makecell[l]{
                            $\varphi_1:(10\le\!\textit{lev}_0\!\le\! 95)\!\wedge\!(10\!\le\!\textit{lev}_1\!\le 95)\!\wedge\!(|\textit{lev}_0 \!-\! \textit{lev}_1|\!\le\! 10) \until[50,60](50\!\le \! \textit{lev}_0\!\le\!80)$\\
                            $\qquad\wedge(50\!\le\!\textit{lev}_1\!\le\! 80)$}}\\
                    \specialrule{0em}{1pt}{1pt}
                        & ~ & \scriptsize{\makecell[l]{
                            $\varphi_2:(10\le\!\textit{lev}_0\!\le\! 95)\!\wedge\!(10\!\le\!\textit{lev}_1\!\le 95)\!\wedge\!(|\textit{lev}_0 \!-\! \textit{lev}_1|\!\le \!10) \until[30,40](50\!\le \! \textit{lev}_0\!\le\!80)$\\
                            $\qquad\wedge(50\!\le\!\textit{lev}_1\!\le\! 80)$}}\\
                    \specialrule{0em}{1pt}{1pt}
                        & ~ & \scriptsize{\makecell[l]{
                            $\varphi_3:(10\le\!\textit{lev}_0\!\le\! 95)\!\wedge\!(10\!\le\!\textit{lev}_1\!\le 95)\until[30,40](50\!\le\!\textit{lev}_0\!\le 80)\wedge(50\!\le\!\textit{lev}_1\!\le\! 80)$}}\\
                        \cmidrule{2-3}
                        \multirow{3}{*}{\textsf{CarSeq}\cite{bae2019bounded}} 
                        & ~ & \scriptsize{\makecell[l]{
                            $\varphi_1:(1\!\le\! \textit{pos}_0\!-\!\textit{pos}_1 \!\le \! 3)\until[2,3](20\!\le \! \textit{pos}_0\!\le\! 25)$}}\\
                    \specialrule{0em}{1pt}{1pt}
                        & ~ & \scriptsize{\makecell[l]{
                            $\varphi_2:(1\!\le\! \textit{pos}_0\!-\!\textit{pos}_1 \!\le \! 3)\!\wedge\!(1\!\le\!\textit{pos}_1\!-\!\textit{pos}_2)\until[2,3]\,(20\!\le\!\textit{pos}_0\!\le\! 25)$}}\\
                    \specialrule{0em}{1pt}{1pt}
                        & ~ & \scriptsize{\makecell[l]{
                            $\varphi_3:(1\!\le\! \textit{pos}_0\!-\!\textit{pos}_1 \!\le \! 3)\!\wedge\!(1\!\le\!\textit{pos}_1\!-\!\textit{pos}_2\!\le\!3 )\!\wedge\!(1\!\le\!\textit{pos}_2\!-\!\textit{pos}_3)\until[2,3]$\\
                            $\qquad(20\!\le\!\textit{pos}_0\!\le\! 25)$}}\\
                    \cmidrule{2-3}
                        \textsf{Oscillator}\cite{xue2023reach} 
                        & ~ & \scriptsize{\makecell[l]{
                            $\varphi~:(\textit{x}^2\!+\!\textit{y}^2\!\le\!1) \until[3,4] (\textit{x}^2\!+\! \textit{y}^2 \!\le\! 0.01)$}}\\
                    \cmidrule{2-3}
                        \multirow{3}{*}{\textsf{Temperature}\cite{bae2019bounded}} 
                        & ~ & \scriptsize{\makecell[l]{
                            $\varphi_1:\wedge_{i=1\!,2\!,3}(23\!\le\! \textit{temp}_i\!\!\le \! 29)\until[8,10]\!\wedge_{i=1\!,2\!,3}\!(26\!\le\!\textit{temp}_i\!\le\! 28)$}} \\
                    \specialrule{0em}{1pt}{1pt}
                        & ~ & \scriptsize{\makecell[l]{
                            $\varphi_2:\wedge_{i=1\!,2\!,3}(23\!\le\! \textit{temp}_i\!\!\le \! 29)\until[8,10]\!\wedge_{i=1\!,2\!,3}\!(26\!\le\!\textit{temp}_i\!\le\! 28)\!\wedge\!(\textit{temp}_2\!\le\!\textit{temp}_1)$}}\\
                    \specialrule{0em}{1pt}{1pt}
                        & ~ & \scriptsize{\makecell[l]{
                            $\varphi_3:\wedge_{i=1\!,2\!,3}(23\!\le\! \textit{temp}_i\!\!\le \! 29)\until[8,10]\!\wedge_{i=1\!,2\!,3}\!(26\!\le\!\textit{temp}_i\!\le\! 28)\!\wedge\!(\textit{temp}_2\!\le\!\textit{temp}_1)$\\
                            $\qquad\wedge (\textit{temp}_3\!\le\!\textit{temp}_2)$}}\\
                    \bottomrule
                \end{tabular}
            \end{center}
            \vspace*{-\baselineskip}
            \vspace*{1mm}
            \scriptsize{
                More detail explanation of the ST-RA formula can be found in \cref{appendix:detail}.    
            } 
            \vspace*{-7mm}
        \end{table*}

        \begin{table*}[t]
            \centering
                %\vspace{-1cm}
                \captionsetup{font={scriptsize}}
                \caption{Empirical results on benchmark examples }
                \vspace{-0.3cm}
                \label{tab:result}
                \begin{center}
                    \begin{tabular}{lcccc c ccc c lcc} 
                        \toprule
                        \multirow{2}{*}[-0.5ex]{Model} &~& \multirow{2}{*}[-0.5ex]{Dynamics}& ~ & \multirow{2}{*}[-0.5ex]{ST-RA} & ~ & \multicolumn{3}{c}{Model Scale} & ~ & \multicolumn{3}{c}{Synthesis Time} \\
                          \cmidrule{7-9}\cmidrule{11-13}
                           & & & & & ~ & $ n_{dim} $ & ~ & $n_{mode}$ & ~ & \#Iter. & ~ & Time (s) \\
                        \midrule
                        \multirow{3}{*}{\textsf{Reactor}\cite{zhao2013synthesizing}}& ~ &\multirow{3}{*}{Const} 
                            & &     $\varphi$ &~& 2 &~& 4 &~& 6 \fp &~& 0.31  \\
                            & & & & $\varphi$ &~& 2 &~& 8 &~& 6 \fp &~& 4.14  \\
                            & & & & $\varphi$ &~& 2 &~& 10 &~& 6 \fp &~& 8.01 \\
                        \cmidrule{3-13}
                        \multirow{3}{*}{\makecell[l]{\textsf{WaterTank}\cite{raisch1999approximating}}} & ~ & \multirow{3}{*}{Const} 
                            & &     $\varphi_1$ &~& 2 &~& 7 &~& 9 \fp &~& 18.04  \\
                            & & & & $\varphi_2$ &~& 2 &~& 7 &~& 6 \fp &~& 10.63  \\
                            & & & & $\varphi_3$ &~& 2 &~& 7 &~& 6 \fp &~& 5.24  \\
                        \cmidrule{3-13}
                        \multirow{3}{*}{\textsf{CarSeq}\cite{bae2019bounded}} & ~ & \multirow{3}{*}{Const} 
                            & &     $\varphi_1$ &~& 2 &~& 4  &~& 5 \fp &~& 1.12  \\
                            & & & & $\varphi_2$ &~& 3 &~& 8  &~& 7 \fp &~& 47.41  \\
                            & & & & $\varphi_3$ &~& 4 &~& 16 &~& 4 &~& 134.79 \\
                        \cmidrule{3-13}
                        \multirow{3}{*}{\textsf{Oscillator\cite{xue2023reach}}} & ~ & \multirow{3}{*}{Poly}
                            & &     $\varphi$ &~& 2 &~& 3  &~& 6 &~& 77.20 \\
                            & & & & $\varphi$ &~& 2 &~& 4  &~& 6 &~& 106.09 \\
                            & & & & $\varphi$ &~& 2 &~& 5 &~& 6 &~& 155.77 \\
                        \cmidrule{3-13}
                        \multirow{3}{*}{\textsf{Temperature}\cite{bae2019bounded}} & ~ & \multirow{3}{*}{Linear}
                            & &     $\varphi_1$ &~& 3 &~& 8  &~& 5 &~& 236.99 \\
                            & & & & $\varphi_2$ &~& 3 &~& 8  &~& 5 &~& 293.66 \\
                            & & & & $\varphi_3$ &~& 3 &~& 8  &~& 5 &~& 252.32 \\
                        \bottomrule
                    \end{tabular}
                \end{center}
                \vspace*{-\baselineskip}
                \vspace*{2mm}
                 
            \scriptsize{
                Dynamics: the type of continuous dynamics; ST-RA: formulas to be satisfied (cf. \cref{tab:explain});\\ 
                $n_{dim}$: dimension of state; $n_{mode}$: number of modes; \#Iter.: number of iterations, $\fp$ means the synthesized set $X_q^i$ (cf. \cref{sec:prob1}) reach a fixpoint at current iteration.
            }
            \vspace*{-6mm}
        \end{table*}

    As shown in \cref{tab:explain}, our experiments involve five distinct models, with three exhibiting constant dynamics and two exhibiting non-constant dynamics. We adjust the model scale or the ST-RA formula for each model to assess the efficacy of our method under varying conditions. In total, 15 different benchmarks are included in our study. \cref{tab:result} details the empirical results of the benchmarks. In each case, the synthesis process continue iterating until either a fixpoint is achieved or the maximum calculation time of 5 minutes is met. 

    Our empirical results illustrate that our method is capable of effectively synthesizing controllers for models with both constant and non-constant dynamics. Notably, for models with constant dynamics, the iterative process tends to converge to a fixpoint, meaning that a complete controller is achieved. Moreover, the synthesis time for these controllers is significantly influenced by both the scale of the model and the complexity of ST-RA formulas. Specially, our analysis reveals:
    \begin{enumerate*}[label=(\roman*)]
        \item an increased number of modes (\textsf{Reactor}) or a higher state dimension (\textsf{CarSeq}) both lead to prolonged synthesis times,
        \item more intricate predicates or larger future-reach time\footnote{Future-reach time refers to the maximum time horizon required to verify the correctness of an STL formula \cite{bae2019bounded}, in \textsf{WaterTank}, the future-reach times of $\varphi_1$, $\varphi_2$, and $\varphi_3$ are $60$, $40$, and $40$ respectively.} (\textsf{WaterTank}) results in increased synthesis times.
    \end{enumerate*}
    For the third benchmark within \textsf{CarSeq}, the model does not reach a fixpoint, primarily because the large model scale rapidly increase the formula size, posing substantial challenges for the Z3 solver.

    When dealing with non-constant dynamics, an approximation method is applied, thereby a fixpoint might not be achievable. The influence of model scale on synthesis time remains consistent with that observed in constant ODE models, as evidenced in \textsf{Oscillator}. Interestingly, the synthesis time for controllers using approximation methods is less affected by the complexity of the ST-RA formula. For example, in \textsf{Temperature}, despite $\varphi_3$ being more complex than $\varphi_2$, it requires less synthesis time, primarily because the complexity of SDP is influenced more by state space dimensions than by constraints.  

    Overall, our method exhibits a high capability in synthesizing switching controller for HSs with various dynamics. It can achieve sound and complete results for constant dynamics within a reasonable time. For more general dynamics, our method can still synthesize a sound result in a reasonable time.

\section{Related Work}\label{sec:related}

    %\hscommentinline{Add the discussion about relation of our method with backwards induction of timed gaming }

    HSs have been a key research focus in the academic community\cite{witsenhausen1966class}. The autonomous \emph{verification} and \emph{synthesis} of HSs began from timed automata \cite{alur1994theory}. Subsequently, various mathematical models, including hybrid automata \cite{henzinger1996theory,henzinger2000symbolic} and various types of differential equations, have been employed to reason about HSs. For a survey of these methods, refer to \cite{deshmukh2019formal}.

    In the realm of formal synthesis of HSs, different methods can be classified along several dimensions.
    \begin{enumerate*}[label=(\roman*)]
        \item Along the designable part of the system, the synthesis problem can be categorized into \emph{feedback controller} synthesis \cite{sanfelice2021hybrid}, \emph{switching controller} synthesis \cite{liberzon2003switching,jha2011synthesis}, and \emph{reset controller} synthesis \cite{clegg1958nonlinear}. 
        \item Along the properties of interest, the problem can be classified into \emph{safety} controller synthesis, \emph{liveness} controller synthesis, and so on.  
    \end{enumerate*}

    Switching controller synthesis \cite{liberzon2003switching}, shaping HSs by strategically constraining their discrete behavior, can be categorized into two fundamentally approaches. The first is based on constraint solving \cite{taly2011synthesizing,zhao2013synthesizing}. This approach highly dependents on finding suitable certificate templates, which is challenging to generate manually. The other approach is abstraction-based method. Given its capability to easily handle complex temporal specifications, this method has been increasingly adopted in recent research \cite{liu2013synthesis,aydin2012language,girard2012controller}.
    %The synthesis of HSs based on temporal logic can be traced back to the seminal work in \cite{tabuada2003model}, which proposed a method by abstracting the discrete-time linear system to a transition system to analyze the system's behavior with respect to linear temporal logic (LTL). Several works in control synthesis followed this line, such as \cite{runggercontroller,wolff2016optimal}.

    The synthesis of HSs concerning reach-avoid type specifications, similar to those discussed in this paper, predominantly focuses on feedback controllers. Notable methods include the Counterexample-Guided Inductive Synthesis (CEGIS) approach proposed by Hadi and Sriram \cite{ravanbakhsh2015counterexample,ravanbakhsh2016robust}, optimization-based methods \cite{xue2023reach}, and others \cite{franzle2019memory,prajna2007convex}.

    When considering STL as the specification, most works have focused solely on the continuous dynamics of HSs. Raman et al. proposed a method to encode the STL specification of a hybrid system into Mixed Integer Linear Programming (MILP) \cite{raman2014model}. This method was employed to synthesize a robust controller in a CEGIS manner in \cite{raman2015reactive}. Synthesizing a controller by reinforcement learning technique for an essential discrete-time system is also introduced recently \cite{meng2023signal}. The Control Barrier Function-based method can also be used to synthesize feedback controller with respect to STL, without requiring discretization of the continuous system \cite{lindemann2018control}. 
    %Using a candidate time-varying barrier function \cite{igarashi2019time}, the controller that renders the system to satisfy the STL specification can be computed by quadratic optimization. 
    While \cite{da2021symbolic} is the only work we know that considers synthesizing switched systems with respect to STL specification, the synthesized part is the switch input for the hybrid automata with discrete dynamics. In contrast, our work is aimed at synthesizing switching controllers that determine the switch time for the system.
    
    Although numerous studies \cite{ravanbakhsh2015counterexample,ravanbakhsh2016robust} address the reach-avoid type specifications of hybrid systems discussed in this paper, the majority of them focus on feedback controllers rather than switching controllers.
\section{Conclusion}\label{sec:conclu}

    We proposed a novel method to synthesize switching controllers for HSs against a fragment of the STL. Our method iteratively calculates the state-time set for each mode, which services as foundation of the synthesize algorithm. The distinctive feature of our approach lies in its soundness and relative completeness. %ensuring that these sets enable the resulting executions to meet the STL specification within a certain number of switches. The results of these calculations are instrumental in deriving switching controller. 
    Our preliminary experiments, leveraging a range of notable examples from existing literature, have effectively demonstrated the method's efficiency and efficacy.

    For future work, we plan to continue to explore in two directions.
    \begin{enumerate*}[label = (\roman*)]
        \item Enlarge the range of specifications under consideration to encompass general STL formulas featuring nested temporal operators. The primary challenge here is devising a unified, recursive formula reasoning approach for general STL specifications.
        \item Broaden the types of controllers that can be synthesized from the calculated state-time sets.
%        \item Examine more general hybrid systems that exhibit increasingly complex dynamics in each mode. Given that our current method is contingent on the analytical solutions of ODEs, investigating ways to reduce this dependency presents a significant and worthwhile challenge.
    \end{enumerate*}

\bibliographystyle{abbrv}
\bibliography{reference}

\begin{thebibliography}{10}

\bibitem{alur1994theory}
R.~Alur and D.~L. Dill.
\newblock A theory of timed automata.
\newblock {\em Theoretical computer science}, 126(2):183--235, 1994.

\bibitem{arnon1984cylindrical}
D.~S. Arnon, G.~E. Collins, and S.~McCallum.
\newblock Cylindrical algebraic decomposition i: The basic algorithm.
\newblock {\em SIAM Journal on Computing}, 13(4):865--877, 1984.

\bibitem{atkins2013aerospace}
E.~M. Atkins and J.~M. Bradley.
\newblock Aerospace cyber-physical systems education.
\newblock In {\em AIAA Infotech@ Aerospace (I@ A) Conference}, page 4809, 2013.

\bibitem{aydin2012language}
E.~Aydin~Gol, M.~Lazar, and C.~Belta.
\newblock Language-guided controller synthesis for discrete-time linear
  systems.
\newblock In {\em Proceedings of the 15th ACM international conference on
  Hybrid Systems: Computation and Control}, pages 95--104, 2012.

\bibitem{bae2019bounded}
K.~Bae and J.~Lee.
\newblock Bounded model checking of signal temporal logic properties using
  syntactic separation.
\newblock {\em Proceedings of the ACM on Programming Languages}, 3(POPL):1--30,
  2019.

\bibitem{chen2012taylor}
X.~Chen, E.~Abraham, and S.~Sankaranarayanan.
\newblock Taylor model flowpipe construction for non-linear hybrid systems.
\newblock In {\em 2012 IEEE 33rd Real-Time Systems Symposium}, pages 183--192.
  IEEE, 2012.

\bibitem{chen2013flow}
X.~Chen, E.~{\'A}brah{\'a}m, and S.~Sankaranarayanan.
\newblock Flow*: An analyzer for non-linear hybrid systems.
\newblock In {\em Computer Aided Verification: 25th International Conference,
  CAV 2013, Saint Petersburg, Russia, July 13-19, 2013. Proceedings 25}, pages
  258--263. Springer, 2013.

\bibitem{clegg1958nonlinear}
J.~C. Clegg.
\newblock A nonlinear integrator for servomechanisms.
\newblock {\em Transactions of the American Institute of Electrical Engineers,
  Part II: Applications and Industry}, 77(1):41--42, 1958.

\bibitem{da2021symbolic}
R.~R. da~Silva, V.~Kurtz, and H.~Lin.
\newblock Symbolic control of hybrid systems from signal temporal logic
  specifications.
\newblock {\em Guidance, Navigation and Control}, 1(02):2150008, 2021.

\bibitem{de2003element}
L.~De~Alfaro, M.~Faella, T.~A. Henzinger, R.~Majumdar, and M.~Stoelinga.
\newblock The element of surprise in timed games.
\newblock In {\em International Conference on Concurrency Theory}, pages
  144--158. Springer, 2003.

\bibitem{de2008z3}
L.~De~Moura and N.~Bj{\o}rner.
\newblock Z3: An efficient smt solver.
\newblock In {\em International conference on Tools and Algorithms for the
  Construction and Analysis of Systems}, pages 337--340. Springer, 2008.

\bibitem{deshmukh2019formal}
J.~V. Deshmukh and S.~Sankaranarayanan.
\newblock Formal techniques for verification and testing of cyber-physical
  systems.
\newblock {\em Design Automation of Cyber-Physical Systems}, pages 69--105,
  2019.

\bibitem{engell2000continuous}
S.~Engell, S.~Kowalewski, C.~Schulz, and O.~Stursberg.
\newblock Continuous-discrete interactions in chemical processing plants.
\newblock {\em Proceedings of the IEEE}, 88(7):1050--1068, 2000.

\bibitem{franzle2019memory}
M.~Fr{\"a}nzle, M.~Chen, and P.~Kr{\"o}ger.
\newblock In memory of oded maler: automatic reachability analysis of
  hybrid-state automata.
\newblock {\em ACM SIGLOG News}, 6(1):19--39, 2019.

\bibitem{girard2012controller}
A.~Girard.
\newblock Controller synthesis for safety and reachability via approximate
  bisimulation.
\newblock {\em Automatica}, 48(5):947--953, 2012.

\bibitem{henzinger1996theory}
T.~A. Henzinger.
\newblock The theory of hybrid automata.
\newblock In {\em Proceedings 11th Annual IEEE Symposium on Logic in Computer
  Science}, pages 278--292. IEEE, 1996.

\bibitem{henzinger2000symbolic}
T.~A. Henzinger and R.~Majumdar.
\newblock Symbolic model checking for rectangular hybrid systems.
\newblock In {\em International Conference on Tools and Algorithms for the
  Construction and Analysis of Systems}, pages 142--156. Springer, 2000.

\bibitem{jha2011synthesis}
S.~Jha, S.~A. Seshia, and A.~Tiwari.
\newblock Synthesis of optimal switching logic for hybrid systems.
\newblock In {\em Proceedings of the ninth ACM international conference on
  Embedded software}, pages 107--116, 2011.

\bibitem{liberzon2003switching}
D.~Liberzon.
\newblock {\em Switching in systems and control}, volume 190.
\newblock Springer, 2003.

\bibitem{lindemann2018control}
L.~Lindemann and D.~V. Dimarogonas.
\newblock Control barrier functions for signal temporal logic tasks.
\newblock {\em IEEE control systems letters}, 3(1):96--101, 2018.

\bibitem{lindemann2019coupled}
L.~Lindemann, J.~Nowak, L.~Sch{\"o}nb{\"a}chler, M.~Guo, J.~Tumova, and D.~V.
  Dimarogonas.
\newblock Coupled multi-robot systems under linear temporal logic and signal
  temporal logic tasks.
\newblock {\em IEEE Transactions on Control Systems Technology},
  29(2):858--865, 2019.

\bibitem{liu2013synthesis}
J.~Liu, N.~Ozay, U.~Topcu, and R.~M. Murray.
\newblock Synthesis of reactive switching protocols from temporal logic
  specifications.
\newblock {\em IEEE Transactions on Automatic Control}, 58(7):1771--1785, 2013.

\bibitem{maler2004monitoring}
O.~Maler and D.~Nickovic.
\newblock Monitoring temporal properties of continuous signals.
\newblock In {\em International Symposium on Formal Techniques in Real-Time and
  Fault-Tolerant Systems}, pages 152--166. Springer, 2004.

\bibitem{mazo2010pessoa}
M.~Mazo~Jr, A.~Davitian, and P.~Tabuada.
\newblock Pessoa: A tool for embedded controller synthesis.
\newblock In {\em International conference on computer aided verification},
  pages 566--569. Springer, 2010.

\bibitem{meng2023signal}
Y.~Meng and C.~Fan.
\newblock Signal temporal logic neural predictive control.
\newblock {\em IEEE Robotics and Automation Letters}, 2023.

\bibitem{mosek}
A.~Mosek.
\newblock The {MOSEK} optimization toolbox for {MATLAB} manual. version 7.1
  (revision 28).
\newblock {\em http://mosek. com, (accessed on March 20, 2015)}, 2015.

\bibitem{prajna2007convex}
S.~Prajna and A.~Rantzer.
\newblock Convex programs for temporal verification of nonlinear dynamical
  systems.
\newblock {\em SIAM Journal on Control and Optimization}, 46(3):999--1021,
  2007.

\bibitem{raisch1999approximating}
J.~Raisch, E.~Klein, C.~Meder, A.~Itigin, and S.~O’Young.
\newblock Approximating automata and discrete control for continuous
  systems—two examples from process control.
\newblock In {\em Hybrid Systems V 5}, pages 279--303. Springer, 1999.

\bibitem{raman2014model}
V.~Raman, A.~Donz{\'e}, M.~Maasoumy, R.~M. Murray, A.~Sangiovanni-Vincentelli,
  and S.~A. Seshia.
\newblock Model predictive control with signal temporal logic specifications.
\newblock In {\em 53rd IEEE Conference on Decision and Control}, pages 81--87.
  IEEE, 2014.

\bibitem{raman2015reactive}
V.~Raman, A.~Donz{\'e}, D.~Sadigh, R.~M. Murray, and S.~A. Seshia.
\newblock Reactive synthesis from signal temporal logic specifications.
\newblock In {\em Proceedings of the 18th international conference on hybrid
  systems: Computation and control}, pages 239--248, 2015.

\bibitem{ravanbakhsh2015counterexample}
H.~Ravanbakhsh and S.~Sankaranarayanan.
\newblock Counterexample-guided stabilization of switched systems using control
  lyapunov functions.
\newblock In {\em Proceedings of the 18th International Conference on Hybrid
  Systems: Computation and Control}, pages 297--298, 2015.

\bibitem{ravanbakhsh2016robust}
H.~Ravanbakhsh and S.~Sankaranarayanan.
\newblock Robust controller synthesis of switched systems using counterexample
  guided framework.
\newblock In {\em Proceedings of the 13th International Conference on Embedded
  Software}, pages 1--10, 2016.

\bibitem{sanfelice2021hybrid}
R.~G. Sanfelice.
\newblock {\em Hybrid feedback control}.
\newblock Princeton University Press, 2021.

\bibitem{taly2011synthesizing}
A.~Taly, S.~Gulwani, and A.~Tiwari.
\newblock Synthesizing switching logic using constraint solving.
\newblock {\em International journal on software tools for technology
  transfer}, 13(6):519--535, 2011.

\bibitem{tomlin2000game}
C.~J. Tomlin, J.~Lygeros, and S.~S. Sastry.
\newblock A game theoretic approach to controller design for hybrid systems.
\newblock {\em Proceedings of the IEEE}, 88(7):949--970, 2000.

\bibitem{weispfenning1988complexity}
V.~Weispfenning.
\newblock The complexity of linear problems in fields.
\newblock {\em Journal of symbolic computation}, 5(1-2):3--27, 1988.

\bibitem{witsenhausen1966class}
H.~Witsenhausen.
\newblock A class of hybrid-state continuous-time dynamic systems.
\newblock {\em IEEE Transactions on Automatic Control}, 11(2):161--167, 1966.

\bibitem{xue2019inner}
B.~Xue, M.~Fr{\"a}nzle, and N.~Zhan.
\newblock Inner-approximating reachable sets for polynomial systems with
  time-varying uncertainties.
\newblock {\em IEEE Transactions on Automatic Control}, 65(4):1468--1483, 2019.

\bibitem{xue2016under}
B.~Xue, Z.~She, and A.~Easwaran.
\newblock Under-approximating backward reachable sets by polytopes.
\newblock In {\em Computer Aided Verification: 28th International Conference,
  CAV 2016, Toronto, ON, Canada, July 17-23, 2016, Proceedings, Part I 28},
  pages 457--476. Springer, 2016.

\bibitem{xue2023reach}
B.~Xue, N.~Zhan, M.~Fr{\"a}nzle, J.~Wang, and W.~Liu.
\newblock Reach-avoid verification based on convex optimization.
\newblock {\em IEEE Transactions on Automatic Control}, 2023.

\bibitem{ye2008modelling}
P.~Ye, E.~Entcheva, S.~A. Smolka, and R.~Grosu.
\newblock Modelling excitable cells using cycle-linear hybrid automata.
\newblock {\em IET systems biology}, 2(1):24--32, 2008.

\bibitem{zhao2013synthesizing}
H.~Zhao, N.~Zhan, and D.~Kapur.
\newblock Synthesizing switching controllers for hybrid systems by generating
  invariants.
\newblock {\em Theories of Programming and Formal Methods: Essays Dedicated to
  Jifeng He on the Occasion of His 70th Birthday}, pages 354--373, 2013.

\end{thebibliography}

\appendix
\newpage
\section{Proofs of Lemmas and Theorems}\label{apendix:proof}

\restateAllInitial*

\begin{proof}[proof of \cref{cor:all_initial}]
 (1) By definition of the state-time sets, $X_q^0 \subseteq X_q^1 \subseteq X_q^2 \subseteq \cdots $ trivially holds. (2) For any $x \in  X_i^q [t \repl 0]$, by definition of state-time sets, we have $(\xx,0) \models \phi_1\mathcal{U}_{I}\phi_2$ where $\xx$ is the solution of $\dot{\xx}(t) = f_{\eta(t) } (\xx(t), t) $ over $[0,\infty)$ with $\xx(0) = x$ for some controller $\eta$. (3) This trivially holds from (1).
    \qed
\end{proof}

\bigskip

Before proving \cref{thm:induc_state-time_set}, we first prove a lemma:

\begin{lemma}\label{lem:eq-init-induc}
            For a given ST-RA formula $\varphi=\phi_1\until[I]\phi_2$, a signal $\Traj[]$, and time instant $\tau$,
            \begin{align*}
                \left(\Traj[],\tau \models \phi_1\until_{I\Pminus \tau}\phi_2\right)\iff \left(\Traj[],\tau \models \phi_1\until[](\phi_2\wedge t\in I)\right)
            \end{align*}
        \end{lemma}
        \begin{proof}
        The following equivalent conditions hold:
            \begin{align*}
                           &\Traj[],\tau \models \phi_1\until_{I\Pminus \tau}\phi_2  \\
                \iff & \exists \tau'\ge \tau,~ \tau'-\tau\in I\Pminus \tau,(\Traj[],\tau')\models \phi_2,\text{and } \forall \tau''\in[\tau,\tau'],~ (\Traj[],\tau'')\models\phi_1\\
                \iff &  \exists \tau'\ge \tau,~ \tau'-\tau\in I-\tau\cap\NonNegReals,(\Traj[],\tau')\models \phi_2,\text{and } \forall \tau''\in[\tau,\tau'],~ (\Traj[],\tau'')\models\phi_1\\
                \iff &  \exists \tau'\ge \tau,~ \tau'\in I, (\Traj[],\tau')\models \phi_2,\text{and } \forall \tau''\in[\tau,\tau'],~ (\Traj[],\tau'')\models\phi_1\\
                \iff &  \exists \tau'\ge \tau,~ (\Traj[],\tau')\models (\phi_2\wedge t\in I),\text{and } \forall \tau''\in[\tau,\tau'],~ (\Traj[],\tau'')\models\phi_1\\
                \iff & \Traj[],\tau\models \phi_1\until (\phi_2\wedge t\in I)
            \end{align*}
            This complete the proof. \qed
        \end{proof}

\bigskip

\restateInducStateTimeSet*
\begin{proof}[proof of \cref{thm:induc_state-time_set}]
(1) By definition of the state-time sets, $(x,\tau)\in X_0^q$ iff $(\xx, \tau) \models \phi_1\until_{I\Pminus \tau}\phi_2 $ where $\xx$ is the solution of ODE $\dot{\xx}(t) = f_{q } (\xx(t), t) $ over $[\tau,\infty)$ with $\xx(\tau) = x$. Moreover, by \cref{lem:eq-init-induc}, we have 
\[\phi_1\, \mathcal{U}\, (\phi_2\wedge t\in I) \equiv \phi_1\until_{I}\phi_2 .\] 
The result follows directly from the above two facts.

(2) ``$\Rightarrow$": By definition of the state time set $X_q^i$, if $(x,\tau) \in X_q^i$, then there exists a switching controller $\pi \from [\tau,\infty )\to Q$, such that $\pi(\tau) = q$, and $\pi$ contains at most $i$ discontinuous points. Moreover $(\xx, \tau) \models \phi_1\until_{I\Pminus \tau}\phi_2 $, where $\xx$ is the solution induced by $\pi$ with initial $x$ at time $\tau$. Let 
\[\tau' \defeq \inf \{t\mid \pi(t) \neq q\}\] 
denote the first switching time of $\pi$, $\pi' \defeq \pi |_{[\tau',\infty)}$ denote the restriction of $\eta$ on $[\tau',\infty)$, and $q' \defeq \pi(\tau')$ the mode at time $\tau'$. It suffice to show 
\[(\Psi(\cdot\,; x, \tau, q),\tau) \models  \phi_1\, \mathcal{U}\, X_{q'}^{i-1},\] 
where $\Psi(\cdot\,; x, \tau, q)$
is the solution of $ \dot{\xx}(t) = f_q (\xx(t)) $ with initial $x$ at time $\tau$. 
In fact, since $(\xx, \tau) \models \phi_1\until_{I\Pminus \tau}\phi_2 $  under $\pi$ with $\pi(\tau) = q$, we have 
\[(\Psi(\cdot\,; x, \tau, q),\tau) \models  \phi_1\, \mathcal{U}\, \{ (\Psi(\tau'\,; x, \tau, q), \tau') \}.\] Thus, we only need to show 
\[(\Psi(\tau'\,; x, \tau, q), \tau') \in X_{q'}^{i-1}\]
Since $(\xx, \tau) \models \phi_1\until_{I\Pminus \tau}\phi_2 $, we have, by \cref{lem:eq-init-induc}, $(\xx, \tau) \models \phi_1\until(\phi_2\wedge t\in I) $, this implies 
\[(\Psi(\cdot \,; x, \tau, q), \tau') \models \phi_1\until(\phi_2\wedge t\in I) .\] 
Again, by \cref{lem:eq-init-induc}, we have 
\[(\Psi(\cdot\,; x, \tau, q), \tau') \models \phi_1\until_{I\Pminus \tau'}(\phi_2) .\] 
Moreover, since $\pi'$ is the restriction of $\pi$ over $[\tau',\infty)$, it satisfies $\pi'(\tau')=q'$ and contains at most $i-1$ discontinuous points. This implies $(\Psi(\tau'\,; x, \tau, q), \tau') \in X_{q'}^{i-1}$ and completes the proof.

``$\Leftarrow$": Suppose there exists $q'\neq q$, such that $(\Psi(\cdot \,; x, \tau, q), \tau) \models \phi_1\until X_{q'}^{i-1} $, then there exists time $\tau'$, satisfying
\[ (\Psi(\tau'\,; x, \tau, q), \tau) \in X_{q'}^{i-1}\quad \text{ and } \quad (\Psi(t\,; x, \tau, q), \tau) \models \phi_1 \text{ for any } \tau\leq t \leq \tau'.\]
Since $(\Psi(\tau'\,; x, \tau, q), \tau) \in X_{q'}^{i-1}$, there exists a controller $\pi'\from [\tau',\infty) \to Q$ with at most $i-1$ discontinuous points,  such that $\pi'(\tau')= q'$, and  $(\xx, \tau') \models \phi_1\until_{I\Pminus \tau'}\phi_2 $, where $\xx$ is the solution induced by $\pi'$ with initial $\Psi(\tau'\,; x, \tau, q)$ at time $\tau'$. Let 
\[
\pi(x)(t) = \begin{cases}
    q & \text{if } \tau\leq t < \tau',\\
    \pi'(x)(t) & \text{if } t \geq  \tau'.
\end{cases}
\]
It is direct to check $\pi(x)$ has at most $i$ discontinuous points. Moreover, Since $(\xx, \tau') \models \phi_1\until_{I\Pminus \tau'}\phi_2 $, we have, by \cref{lem:eq-init-induc},
\[
(\xx, \tau') \models \phi_1\until (\phi_2 \wedge t\in I),
\]
thus 
\[
(\xx, \tau) \models \phi_1\until (\phi_2 \wedge t\in I),
\]
holds as $(\Psi(t\,; x, \tau, q), \tau) \models \phi_1$  for any $ \tau\leq t \leq \tau'$, which further implies $(\xx, \tau) \models \phi_1\until_{I\Pminus \tau}\phi_2 $ where $\xx$ is the solution induced by $\pi$ with initial $x$ at time $\tau$. This completes the proof. \qed
\end{proof}

\bigskip

\restateComputeTimeStateSet*
\begin{proof}[proof of \cref{thm:compute_time-state_set}]
    \cref{eq:X_q^0,eq:X_q^i} are direct translation of \cref{eq:induc_base,eq:induc_i}. Taking \cref{eq:X_q^0} as an example, 
    \[\phi_2[x,t\repl\Psi(t+ \delta;x, t, q),t+\delta] \wedge (t + \delta \in I)\]
    encapsulates the condition that the system must fulfills $(\phi_2\wedge t\in I)$ at some time $t+ \delta$, and 
    \[ \forall 0\leq h\leq \delta,\, \phi_1 [x,t\repl\Psi(t+ h;x, t, q),t+h]\]
    encodes the requirement that the system must satisfies $\phi_1$ at every instant within the interval $[t, t+h]$. 
    Similar arguments holds for \cref{eq:X_q^i}. This completes the proof. \qed
\end{proof}

\bigskip

\restateConstantFlow*
\begin{proof}[proof of \cref{cor:constant_flow}]
    Since $f_q = a_q$ for any $q\in Q$, the solution of ODE $ \dot{\xx}(t) = f_q (\xx(t)) $ with initial $x$ at time $\tau$ is solved by $ \Psi(t; x, \tau, q) = x + (t-\tau)\cdot a_q$. Moreover, since $\phi_1$ and $\phi_2$ are Boolean combinations of polynomial inequalities, \cref{eq:X_q^0,eq:X_q^i} are first order formulas
over real fields, which exists an explicit algorithm~\cite{arnon1984cylindrical} to compute $\qe$. \qed
\end{proof}

\bigskip

\restateTheoretic*
\begin{proof}[proof of \cref{thm:theoretic}]
    (1) The soundness follows directly from the definition of $\Init$ and $\pi$ in \cref{alg:syn-ss,alg:time-dependent controller}: for any $x_0\in \Init$, \cref{alg:time-dependent controller} iteratively finds the next state-time set that the system will enter, and extracts the corresponding switching time and discrete mode.

    (2) Suppose $x$ can be driven to satisfy $\varphi$ with controller $\pi$, let $k$ denote the the number of discontinuous points of $\pi(x)$, then the initial set $\Init$ synthesized by \cref{alg:syn-ss} contains $x$, since $\Init$ contains all state that can be driven to satisfy $\varphi$ within $k$ times of switches (cf. \cref{cor:all_initial}).

    (3) Notice \cref{alg:time-dependent controller} finds the smallest $l$, such that there exists $q\in Q$, $\Init(q)^l$ contains $x$ (line \ref{alg:init_mode}), $x$ can be driven to satisfy $\varphi$ with \emph{at least} $l$ times of switching by definition of $\Init(q)^l$. Moreover, the controller $\pi(x)$ synthesized by \cref{alg:time-dependent controller} switches $l$ time, thus $\pi(x)$ features minimal switching time property.
\end{proof}

\section{Detail of Experimental Evaluation}\label{appendix:detail}

    \begin{table*}[t]
        \centering
            \captionsetup{font={scriptsize}}
            \caption{Detail Explanation of ST-RA Specifications}
            \vspace{-0.3cm}
            \label{tab:detail-explain}
            \begin{center}
                \begin{tabular}{l  c p{3.5cm} c p{6.4cm}} 
                    \toprule
                    Model  & ~ & ST-RA Formulas & ~ & Detail Explanation  \\
                    \midrule
                    \textsf{Reactor}\cite{zhao2013synthesizing}
                        & ~ & \scriptsize{\makecell[l]{$
                            \varphi : (10\!\le\! \textit{tempe} \!\le\! 90)\wedge(0\!\le$\\ 
                            $\quad~\textit{cooling}\le\! 1) \until[15,20](40\!\le$\\ 
                            $\quad~\textit{tempe} \!\le\! 50)$} }
                        & ~ & 
                        \scriptsize{
                            \makecell[l]{
                            The reactor's \textbf{temperature} will remain between \textbf{10}\\
                            and \textbf{90}, with \textbf{cooling power} between \textbf{0} and \textbf{1}, \emph{until}\\ 
                            a certain moment between \textbf{15} and \textbf{20}. At that poi-\\
                            nt, the \textbf{temperature} will reach between \textbf{40} and \textbf{50}.}
                            }
                            \\
                    \cmidrule{2-5}
                    \multirow{3}{*}[-1.3cm]{\textsf{WaterTank}\cite{raisch1999approximating}} 
                        & ~ & \scriptsize{\makecell[l]{
                            $\varphi_1:(10\le\!\textit{lev}_0\!\le\! 95)\!\wedge\!(10\!\le$\\
                            $\qquad\textit{lev}_1\!\le 95)\!\wedge\!(|\textit{lev}_0 \!-\! \textit{lev}_1|$\\ 
                            $\qquad\le 10) \until[50,60](50\!\le \! \textit{lev}_0\!\le$\\ 
                            $\qquad 80)\wedge(50\!\le\!\textit{lev}_1\!\le\! 80)$}} 
                        & ~ &  
                        \scriptsize{
                        \makecell[l]{
                            The water \textbf{level} in each of a \textbf{double}-watertanks syst-\\
                            em will be between \textbf{10} and \textbf{95}, with a \textbf{difference} of\\ 
                            less than \textbf{10} in between, \emph{until} a moment between\\ 
                            \textbf{50} and \textbf{60}. At that time, both tanks will have water\\ 
                            \textbf{levels} between \textbf{50} and \textbf{80}.
                        }
                        }\\
                    \specialrule{0em}{0pt}{0pt}
                        & ~ & \scriptsize{\makecell[l]{
                            $\varphi_2:(10\le\!\textit{lev}_0\!\le\! 95)\!\wedge\!(10\!\le$\\
                            $\qquad\textit{lev}_1\!\le 95)\!\wedge\!(|\textit{lev}_0 \!-\! \textit{lev}_1|$\\ 
                            $\qquad\le 10) \until[30,40](50\!\le \! \textit{lev}_0\!\le$\\ 
                            $\qquad 80)\wedge(50\!\le\!\textit{lev}_1\!\le\! 80)$}} 
                        & ~ &  
                        \scriptsize{
                        \makecell[l]{
                            The water \textbf{level} in each of a \textbf{double}-watertanks syst-\\
                            em will be between \textbf{10} and \textbf{95}, with a \textbf{difference} of\\ 
                            less than \textbf{10} in between, \emph{until} a moment between\\ 
                            \textbf{30} and \textbf{40}. At that time, both tanks will have water\\ 
                            \textbf{levels} between \textbf{50} and \textbf{80}.
                        }
                        }\\
                    \specialrule{0em}{0pt}{0pt}
                        & ~ & \scriptsize{\makecell[l]{
                            $\varphi_3:(10\le\!\textit{lev}_0\!\le\! 95)\!\wedge\!(10\!\le$\\
                            $\qquad\textit{lev}_1\!\le 95)\until[30,40](50\!\le$\\
                            $\qquad\textit{lev}_0\!\le 80)\wedge(50\!\le\!\textit{lev}_1\!\le$\\ 
                            $\qquad 80)$}} 
                        & ~ &  
                        \scriptsize{
                        \makecell[l]{
                            The water \textbf{level} in each of a \textbf{double}-watertanks syst-\\
                            em will be between \textbf{10} and \textbf{95}, \emph{until} a moment betw-\\
                            een \textbf{30} and \textbf{40}. At that time, both tanks will have\\ 
                            water \textbf{levels} between \textbf{50} and \textbf{80}.
                        }
                        }\\
                        \cmidrule{2-5}
                        \multirow{3}{*}[-0.8cm]{\textsf{CarSeq}\cite{bae2019bounded}} 
                        & ~ & \scriptsize{\makecell[l]{
                            $\varphi_1:(1\!\le\! \textit{pos}_0\!-\!\textit{pos}_1 \!\le \! 3)\until[2,3]$\\
                            $\qquad(20\!\le \! \textit{pos}_0\!\le\! 25)$}} 
                        & ~ &  
                        \scriptsize{
                        \makecell[l]{
                            In a \textbf{two}-car sequence, the first car will be \textbf{1} to \textbf{3}\\
                            meters \textbf{ahead} of the second until a moment between\\ 
                            \textbf{2} and \textbf{3}. At that time, its \textbf{position} will be between\\ 
                            \textbf{20} and \textbf{25} meters.
                        }
                        }\\
                    \specialrule{0em}{0pt}{0pt}
                        & ~ & \scriptsize{\makecell[l]{
                            $\varphi_2:(1\!\le\! \textit{pos}_0\!-\!\textit{pos}_1 \!\le \! 3)\!\wedge\!(1$\\
                            $\qquad\le\!\textit{pos}_1\!-\!\textit{pos}_2)\until[2,3]\,(20$\\
                            $\qquad\le\!\textit{pos}_0\!\le\! 25)$}} 
                        & ~ &  
                        \scriptsize{
                            \makecell[l]{
                            In a \textbf{three}-cars sequence, the middle car should mai-\\
                            ntain a distance of \textbf{1} to \textbf{3} meters behind the first car,\\
                            and the \textbf{last} car should remain \textbf{1} meter \textbf{behind} the\\ 
                            middle one, until a moment between \textbf{2} and \textbf{3} seconds.\\ 
                            Subsequently, the position of the lead car will range\\ 
                            between \textbf{20} and \textbf{25} meters.
                        }
                        }\\
                    \specialrule{0em}{0pt}{0pt}
                        & ~ & \scriptsize{\makecell[l]{
                            $\varphi_3:(1\!\le\! \textit{pos}_0\!-\!\textit{pos}_1 \!\le \! 3)\!\wedge\!(1$\\
                            $\qquad\le\!\textit{pos}_1\!-\!\textit{pos}_2\!\le\!3 )\!\wedge\!(1\!\le$\\
                            $\qquad\textit{pos}_2\!-\!\textit{pos}_3)\until[2,3]\,(20\!\le$\\
                            $\qquad\textit{pos}_0\!\le\! 25)$}} 
                        & ~ &  
                        \scriptsize{
                            \makecell[l]{
                            In a \textbf{four}-cars sequence, the middle \textbf{two} cars should\\ 
                            maintain a distance of \textbf{1} to \textbf{3} meters behind the pre-\\
                            ceding car, and the \textbf{last} car should remain \textbf{1} meter\\ 
                            \textbf{behind} the one in front, until a moment between \textbf{2}\\ 
                            and \textbf{3} seconds. Subsequently, the position of the lead\\ 
                            car will range between \textbf{20} and \textbf{25} meters.
                        }
                        }\\
                    \cmidrule{2-5}
                        \textsf{Oscillator}\cite{xue2023reach} 
                        & ~ & \scriptsize{\makecell[l]{
                            $\varphi:(\textit{x}^2\!+\!\textit{y}^2\!\le\!1) \until[3,4] (\textit{x}^2\!+\! \textit{y}^2 $\\
                            $\quad~ \le\! 0.01)$}} 
                        & ~ &  
                        \scriptsize{
                        \makecell[l]{
                            In a \textbf{two-dimensional} Van der Pol Oscillator system,\\ 
                            the oscillator remains within the \textbf{unit circle} until a \\
                            moment between \textbf{3} and \textbf{4} seconds. At that time, the\\ 
                            position of the oscillator will be in a circle centered at\\ 
                            the origin with a radius of \textbf{0.1}.
                        }
                        }\\
                    \cmidrule{2-5}
                        \multirow{3}{*}[-0.8cm]{\textsf{Temperature}\cite{bae2019bounded}} 
                        & ~ & \scriptsize{\makecell[l]{
                            $\varphi_1:\wedge_{i=1\!,2\!,3}(23\!\le\! \textit{temp}_i\!\!\le \! 29)$ \\ 
                            $\qquad\until[8,10]\!\wedge_{i=1\!,2\!,3}\!(26\!\le\!\textit{temp}_i$\\ 
                            $\qquad\le\! 28)$}} 
                        & ~ &  
                        \scriptsize{
                        \makecell[l]{
                            In a \textbf{three-room} temperature control system, the tem-\\
                            perature in each room will range from \textbf{23} to \textbf{29} degrees\\ 
                            Celsius until a moment between \textbf{8} and \textbf{10} seconds. At\\ 
                            that time, the temperature in all three rooms will be be-\\
                            tween \textbf{26} and \textbf{28} degrees Celsius.
                        }
                        }\\
                    \specialrule{0em}{0pt}{0pt}
                        & ~ & \scriptsize{\makecell[l]{
                            $\varphi_2:\wedge_{i=1\!,2\!,3}(23\!\le\! \textit{temp}_i\!\!\le \! 29)$ \\ 
                            $\qquad\until[8,10]\!\wedge_{i=1\!,2\!,3}\!(26\!\le\!\textit{temp}_i$\\ 
                            $\qquad\le\! 28)\!\wedge\!(\textit{temp}_2\!\le\!\textit{temp}_1)$}} 
                        & ~ &  
                        \scriptsize{
                            \makecell[l]{
                                In a \textbf{three-room} temperature control system, the tem-\\
                                perature in each room will range from \textbf{23} to \textbf{29} degrees\\ 
                                Celsius until a moment between \textbf{8} and \textbf{10} seconds. At\\ 
                                that time, the temperature in all three rooms will be be-\\
                                tween \textbf{26} and \textbf{28} degrees Celsius, and the temperature\\ 
                                of the first room will be higher than the second one. 
                        }
                        }\\
                    \specialrule{0em}{0pt}{0pt}
                        & ~ & \scriptsize{\makecell[l]{
                            $\varphi_3:\wedge_{i=1\!,2\!,3}(23\!\le\! \textit{temp}_i\!\!\le \! 29)$ \\ 
                            $\qquad\until[8,10]\!\wedge_{i=1\!,2\!,3}\!(26\!\le\!\textit{temp}_i$\\ 
                            $\qquad\le\! 28)\!\wedge\!(\textit{temp}_2\!\le\!\textit{temp}_1)$\\
                            $\qquad\wedge (\textit{temp}_3\!\le\!\textit{temp}_2)$}} 
                        & ~ &  
                        \scriptsize{
                            \makecell[l]{
                                In a \textbf{three-room} temperature control system, the tem-\\
                                perature in each room will range from \textbf{23} to \textbf{29} degrees\\ 
                                Celsius until a moment between \textbf{8} and \textbf{10} seconds. At\\ 
                                that time, the temperature in all three rooms will be be-\\
                                tween \textbf{26} and \textbf{28} degrees Celsius, and the temperature\\ 
                                of the first room will be higher than the second one, the\\ 
                                temperature of the second room will be higher than the\\ 
                                last one. 
                        }
                        }\\
                    \bottomrule
                \end{tabular}
            \end{center}
            \vspace*{-\baselineskip}
            \vspace*{-6mm}
            %\scriptsize{
                %time1: time of synthesizing switched system (\cref{alg:time-relevant-syn}); time2: time of synthesizing switched hybrid automaton (\cref{alg:switch-syn}) 
              %  } %\vspace*{-5mm}
        \end{table*}

\end{document}